\documentclass[a4paper,11pt]{amsart}
\usepackage{dsfont,amscd,amssymb}
\usepackage{array}
\usepackage{float}
\usepackage{cmap} 
\usepackage[utf8]{inputenc} 
\usepackage[english]{babel}
\usepackage{stmaryrd}
\usepackage[all]{xy}
\usepackage[pdfdisplaydoctitle=true,
      colorlinks=true,
      urlcolor=blue,
      citecolor=blue,
      linkcolor=blue,
      pdfstartview=FitH,
      pdfpagemode= UseNone,
      bookmarksnumbered=true]{hyperref}

\newtheorem{thm}{Theorem}[section]

\newtheorem{lem}[thm]{Lemma}

\newtheorem{prop}[thm]{Proposition}

\theoremstyle{definition}
\newtheorem{defn}[thm]{Definition}

\theoremstyle{remark}
\newtheorem{rem}[thm]{Remark}
\newtheorem{ex}[thm]{Example}

\newcommand{\NN}{\mathbb{N}}                
\newcommand{\ZZ}{\mathbb{Z}}                
\newcommand{\RR}{\mathbb{R}}                
\newcommand{\CC}{\mathbb{C}}                
\newcommand{\cK}{\Bbbk}                     

\newcommand{\SL}{\mathrm{SL}}               
\newcommand{\SO}{\mathrm{SO}}               
\newcommand{\OO}{\mathrm{O}}                
\newcommand{\SU}{\mathrm{SU}}               

\newcommand{\slc}{\mathfrak{sl}}            

\newcommand{\Ela}{\mathbb{E}\mathrm{la}}    
\newcommand{\ST}[1]{\mathbb{S}_{#1}}        
\newcommand{\HT}[1]{\mathbb{H}_{#1}}        

\newcommand{\Sn}[1]{\mathcal{S}_{#1}}	    
\newcommand{\Hn}[1]{\mathcal{H}_{#1}}	    

\newcommand{\inv}[1]{\mathbf{Inv}(#1)}	    
\newcommand{\cov}[1]{\mathbf{Cov}(#1)}	    

\newcommand{\uu}{\pmb{u}}                   
\newcommand{\vv}{\pmb{v}}                   
\newcommand{\ww}{\pmb{w}}                   
\newcommand{\xx}{\pmb{x}}                   
\newcommand{\bxi}{\pmb{\xi}}                
\newcommand{\yy}{\pmb{y}} 				

\newcommand{\ba}{\mathbf{a}} 				
\newcommand{\bb}{\mathbf{b}} 				
\newcommand{\bc}{\mathbf{c}} 				
\newcommand{\bd}{\mathbf{d}} 				
\newcommand{\bC}{\mathbf{C}} 				
\newcommand{\bD}{\mathbf{D}} 				
\newcommand{\bT}{\mathbf{T}} 				
\newcommand{\bH}{\mathbf{H}} 				

\newcommand{\qq}{\mathrm{q}} 				
\newcommand{\bp}{\mathrm{p}} 				
\newcommand{\bh}{\mathrm{h}} 				

\newcommand{\ff}{\mathbf{f}} 				
\newcommand{\bg}{\mathbf{g}} 				
\newcommand{\hh}{\mathbf{h}} 				
\newcommand{\kk}{\mathbf{k}} 				

\newcommand{\trans}[3]{( #1,#2 )_{#3}}              

\DeclareMathOperator{\Ad}{Ad}
\DeclareMathOperator{\tr}{tr}

\DeclareMathOperator{\dev}{dev}

\newcommand{\norm}[1]{\lVert#1\rVert}       
\newcommand{\set}[1]{\left\{#1\right\}}     

\hypersetup{
  pdfauthor={M. Olive, B. Kolev and N. Auffray}, 
  pdftitle={A minimal integrity basis for the elasticity tensor}, 
  pdfsubject={MSC 2010: 74E10 (15A72 74B05)}, 
  pdfkeywords={Elasticity tensor; Classical Invariant Theory; Integrity Basis; Gordan Algorithm}, 
  pdflang=en, 
  }

\begin{document}

\title{A minimal integrity basis for the elasticity tensor}%

\author[M. Olive]{M. Olive}
\address{LMT-Cachan (ENS Cachan, CNRS, Universit\'{e} Paris Saclay), F-94235 Cachan Cedex, France}
\email{marc.olive@math.cnrs.fr}

\author[B. Kolev]{B. Kolev}
\address{Aix Marseille Universit\'{e}, CNRS, Centrale Marseille, I2M, UMR 7373, 13453 Marseille, France}
\email{boris.kolev@math.cnrs.fr}

\author[N. Auffray]{N. Auffray}
\address{MSME, Universit\'{e} Paris-Est, Laboratoire Mod\'{e}lisation et Simulation Multi Echelle, MSME UMR 8208 CNRS, 5 bd Descartes, 77454 Marne-la-Vall\'{e}e, France}
\email{Nicolas.auffray@univ-mlv.fr}


\date{\today}%

\begin{abstract}
  We definitively solve the old problem of finding a minimal integrity basis of polynomial invariants of the fourth-order elasticity tensor $\bC$. Decomposing $\bC$ into its $\SO(3)$-irreducible components we reduce this problem to finding joint invariants of a triplet $(\ba, \bb, \bD)$, where $\ba$ and $\bb$ are second-order harmonic tensors, and $\bD$ is a fourth-order harmonic tensor. Combining theorems of classical invariant theory and formal computations, a minimal integrity basis of $297$ polynomial invariants for the elasticity tensor is obtained for the first time.
\end{abstract}

\maketitle

\begin{scriptsize}
  \setcounter{tocdepth}{2}
  \tableofcontents
\end{scriptsize}
\clearpage

\section{Introduction}
\label{sec:intro}

In solids mechanics when the matter is slightly deformed the local state of strain is modelled, at each material point, by a second-order symmetric tensor $\boldsymbol{\varepsilon}$. The local stress resulting to the imposed strain is classically described by another second-order symmetric tensor, the Cauchy stress $\boldsymbol{\sigma}$. The way stress and strain are related is defined by a constitutive law. According to the intensity of strain, the nature of the material, and external factors such as the temperature, the nature and type of constitutive laws can vary widely~\cite{LC1994,FPZ1998}.

Among constitutive laws, linear elasticity is one of the simplest model. It supposes a linear relationship between the strain and the stress tensor \emph{at each material point}, $\boldsymbol{\sigma} = \bC: \boldsymbol{\varepsilon}$, in which $\bC$ is a fourth-order tensor, element of a 21-dimensional vector space $\Ela$ \cite{Gur1973,Cow1989,FV1996,AO2008}. From a physical point a view, this relation, which is the 3D extension of the Hooke's law for a linear spring: $F=k\Delta x$, encodes the elastic properties of a body in the small perturbation hypothesis~\cite{Sal2012}.

Due to the existence of a micro-structure at a scale below the one used for the continuum description, elastic properties of many homogeneous materials are anisotropic, i.e. they vary with  material directions. Elastic anisotropy is very common and can be encountered in natural materials (rocks, bones, crystals, \ldots) as well as in manufactured ones (composites, textiles, extruded or rolled irons, \ldots)~\cite{Boe1987,AO2008,Cow2013}. Measuring and modelling the elastic anisotropy of materials is of critical importance for a large kind of applications ranging from the anisotropic fatigue of forged steel \cite{PMM2009}, the damaging of materials \cite{Fra2012,DD2016} to the study of wave propagation in complex materials such as bones \cite{ACVBR1984,RNN2016} or rocks \cite{Bac1970,HT2005}. More recently, the development of acoustic and elastic meta-materials and the wish to conceive paradoxical materials gave a new impulse for the study of anisotropic elasticity \cite{Mil2002,He2004,AO2008,RA2016}.

Working with elastic materials imply the need to identify and distinguish them. A natural question is ``\emph{How to give different names to different elastic materials ?}''. Despite its apparent simplicity, this question formulated for 3D elastic media is a rather hard problem to solve. An elasticity tensor $\bC$ represents a homogeneous material in a specific orientation with respect to a \emph{fixed frame} and a rotation of the body results in another elasticity tensor $\overline{\bC}$ representing the same material. Each homogeneous material is characterized by many elasticity tensors and coordinate-based designation clearly cannot label elastic materials uniquely.

From a mathematical point of view, the material change of orientation makes $\bC$ move in $\Ela$. Classifying anisotropic materials is amount to describe the orbits of the action of the rotation group $\SO(3,\RR)$ on $\Ela$. This can be achieved by determining a finite system of \emph{invariants} which \emph{separates} the orbits.

The analog problem in plane elasticity for the elasticity tensor in bi-dimensional space under the action of the orthogonal group $\OO(2)$ has already been solved by numerous authors~\cite{Ver1982,Gre1995,BOR1996,Via1997,FV2014,AR2016}).

The problem in 3D is much more complicated. The first attempt to define such intrinsic parameters goes back to the seminal work of Lord Kelvin \cite{Tho1856}, rediscovered later by Rychlewski \cite{Ryc1984} and followed since then by many authors \cite{Bet1987,MC1990,Xia1997,Ost1998,BBS2008}. It is based on the representation of the elasticity tensor as a symmetric second-order tensor in $\RR^6$ and the use of its spectral decomposition. However, even if the six eigenvalues of this second-order tensor are invariants, they do not separate the orbits. Worse, the geometry of the problem, which is based on the group $\SO(3,\RR)$ and not $\SO(6,\RR)$ is lost.

The approach we adopt in this paper is somehow different and relies on \emph{representation theory}~\cite{IG1984,Ste1994,GW2009} of the rotation group $\SO(3,\RR)$. It seems to have been pointed out first by Boelher, Kirillov and Onat~\cite{BKO1994} and was already used to describe the symmetry classes using relations between polynomial invariants~\cite{AKP2014a}.

The problem of finding out $\SO(3,\RR)$-invariants of a fourth-order tensor is not new, and has been investigated by many authors (for e.g \cite{Hah1974,Ver1982,Bet1987,Tin1987,Xia1997,Nor2007,Iti2015}). Generally, these invariants are computed using traces of tensor products \cite{Boe1987,Spe1971} and the method relies on some tools developed by Rivlin and others~\cite{SR1958/1959,SR1962,Smi65} for a family of second-order tensors.

There is also a wide literature concerning \emph{separating sets} (also known as \emph{functional bases})~\cite{WP1964,Wan1970,Smi1970}. However, no complete system of separating invariants for the elasticity tensor have been obtained so far. Most of the results exhibit only separating sets for some specific ``generic'' tensors~\cite{BKO1994,Ost1998} or tensors in a given symmetry class~\cite{AKP2014a}. It is also worth emphasizing that a \emph{local system} of coordinates on the orbit space (build up of 18 \emph{locally separating} invariants) should never be confused with a functional basis of $N$ invariants (which may be assimilated to a \emph{global system of parameters}, since they can be used to embed the orbit space in some $\RR^{N}$) \cite{BKO1994,BBS2008}. It is highly improbable that a (global) separating set of 18 (polynomial, rational or algebraic) invariants exists.

A finite set of polynomials generating the algebra of $\SO(3,\RR)$-invariant polynomials is called an \emph{integrity basis}. Note that an integrity basis is always a functional basis (for a real representation of a compact group), but the converse is generally false and it is usually expected to find a functional basis with fewer elements than a minimal integrity basis~\cite{Spe1971,Boe1978}. Integrity bases for the elasticity tensor have already been considered in the literature~\cite{BKO1994,BH1995,SB1997} but results are either incomplete or conjectural.

The main result of this paper, formulated as Theorem~\ref{thm:main-result}, is the determination, for the first time, of a \emph{complete and minimal integrity} basis of 297 polynomials for the elasticity tensor. Although, the theoretical tools to solve this question exist, since at least a hundred years, the effective resolution turns out to be highly complex in practice and has not been solved until now. Even if the exact size of an integrity basis for $\Ela$ was unknown, it was expected to be very large \cite{BKO1994,BH1995,Xia1997}, precluding its determination by hands. Note, furthermore, that the results presented here are not just bounded to questions in continuum mechanics but are also related to other fields such as quantum computation~\cite{Luq2007} and cryptography~\cite{LR2012}.

The computation of the integrity basis requires first the decomposition of the space $\Ela$ into irreducible factors under the $\SO(3,\RR)$-action. This decomposition, known in the mechanics community as the \emph{harmonic decomposition}, results in the splitting of the elasticity tensor $\mathbf{C}$ into five pieces $(\lambda, \mu, \mathbf{a}, \mathbf{b}, \mathbf{D})$, where $\lambda, \mu$ are scalars, $\mathbf{a}, \mathbf{b}$ are deviators and $\mathbf{D}$ is a totally symmetric, traceless fourth order tensor.

Although integrity bases for invariant algebras of each individual irreducible factor $\lambda$, $\mu$, $\ba$, $\bb$, $\bD$ (called \emph{simple invariants}) were already known~\cite{BKO1994,SB1997}, it was still an open question, until now, to determine a full set of \emph{joint invariants} (involving several factors), which, together with the simple invariants, form a minimal integrity basis for the polynomial invariant algebra of $\Ela$.

To compute these joint invariants, we have used a link between \emph{harmonic tensors} in $\RR^{3}$ and \emph{binary forms} (homogeneous polynomials in $\CC^{2}$), reducing the problem to \emph{classical invariant theory}~\cite{Der1999,KP2000,DK2008,Stu2008,DK2015} and allowing to apply Gordan's algorithm~\cite{Gor1868,GY2010} to produce a generating set. This algorithm, which is effective, is the core to make explicit calculations in this field.

\subsection*{Organization}

The paper is organized as follows. In section~\ref{sec:elasticity-tensor}, we introduce the elasticity tensor and the $\SO(3,\RR)$ representation on $\Ela$. In section~\ref{sec:harmonic-decomposition}, we introduce the harmonic decomposition of the elasticity tensor. Basic material about polynomial invariants of the elasticity tensor are recalled in section~\ref{sec:polynomial-invariants}. The reduction of the problem of computing an integrity basis to classical invariant theory is detailed in section~\ref{sec:harmonic-tensor-versus-binary-forms}. Section~\ref{sec:covariants} is principally devoted to explain the main tool that we have used, namely the \emph{Gordan algorithm}. The explicit results are presented in section~\ref{sec:explicit-computations}. Besides, two appendices are provided, one on the harmonic decomposition of a general homogeneous polynomial (and hence a symmetric tensor), and one on the \emph{Cartan map}, used to build an explicit and equivariant isomorphism between the space of harmonic tensors of order $n$ and the space of binary forms of degree $2n$.

\subsection*{Notations}

In the following $\cK$ indicates a field that can be either $\RR$ or $\CC$. The following spaces will be involved:
\begin{itemize}
  \item $\ST{n}(\cK^{3})$ the space of $n$-th order \emph{totally symmetric tensors} on $\cK^{3}$;
  \item $\HT{n}(\cK^{3})$ the space of \emph{harmonic tensors} of order $n$;
  \item $\cK[V]$ the space of polynomial functions (with coefficients in $\cK$) on the vector space $V$;
  \item $\cK_{n}[V]$ the finite-dimensional sub-space of homogeneous polynomials of degree $n$ on $V$;
  \item $\Hn{n}(\cK^{3})$ the space of \emph{harmonic polynomials} of degree $n$;
  \item $\Sn{n}$ the space of binary forms of degree $n$;
  \item $M_{d}(\cK)$ the space of $d$ dimensional square matrices over $\cK$.
\end{itemize}
In addition, we will adopt the following conventions:
\begin{itemize}
  \item $\gamma$ will be an element in $\SL(2,\CC)$;
  \item $g$ is an element of $\SO(3,\cK)$;
  \item $\bxi = (u,v)$ stands for a vector in $\CC^{2}$;
  \item $\vv = (x,y,z)$ stands for a vector in $\CC^{3}$ or $\RR^{3}$;
  \item $\ba, \bb, \bc, \bd$ are second-order tensors;
  \item $\bC, \bD$ are fourth-order tensors;
  \item $\bT$ is a generic tensor;
  \item $\bH$ is a generic harmonic tensor;
  \item $\bp$ is a polynomial on $\cK^{3}$;
  \item $\bh$ is a harmonic polynomial on $\cK^{3}$;
  \item $\ff, \bg, \hh, \kk$ are binary forms.
\end{itemize}

\section{The elasticity tensor and classification of materials}
\label{sec:elasticity-tensor}

In the infinitesimal theory of elasticity \cite{Gur1964}, the \emph{strain tensor} $\boldsymbol{\varepsilon}$ is defined, in Cartesian coordinates, as
\begin{equation*}
  \varepsilon_{ij} = \frac{1}{2}\,\left( \frac{\partial u_{i}}{\partial x^{j}} + \frac{\partial u_{j}}{\partial x^{i}}\right),
\end{equation*}
in which $\uu$ is the displacement field. Classically, internal forces are represented by a contravariant symmetric tensor field, $\boldsymbol{\sigma} = (\sigma^{ij})$, the \emph{Cauchy stress tensor}, and defined at each point of the material. In linear elasticity, the Cauchy stress tensor and the infinitesimal strain tensor are related by the \emph{generalized Hooke's law}
\begin{equation*}
  \sigma^{ij} = C^{ijkl}\,\varepsilon_{kl},
\end{equation*}
where the \emph{elasticity tensor} $\bC = (C^{ijkl})$ is a fourth-order tensor with index symmetry, called the \emph{minor symmetry}
\begin{equation}\label{eq:small-sym}
  C^{ijkl} = C^{jikl} = C^{ijlk}.
\end{equation}
In the case of \emph{hyper-elastic materials}, we have moreover the so-called \emph{major symmetry}
\begin{equation}\label{eq:big-sym}
  C^{ijkl} = C^{klij}.
\end{equation}
We define the space $\Ela$ as the $21$ dimensional vector space of fourth order tensors with index symmetries \eqref{eq:small-sym} and \eqref{eq:big-sym}.

A \emph{homogeneous material} is one for which the tensor field $\bC$ is constant. Thus, to each homogeneous material corresponds an elasticity tensor $\bC$ in $\Ela$, but this correspondence is not unique. Taking another orientation of the material within a \emph{fixed reference frame} corresponds to some transformation $g\in \SO(3,\RR)$. This rotation induces a transformation of the original elasticity tensor $\bC$:
\begin{equation}\label{equ:group_action}
  C^{ijkl} \mapsto g^{i}_{p}\,g^{j}_{q}\,g^{k}_{r}\,g^{l}_{s}\, C^{pqrs},
\end{equation}
where $\mathbf{g} = (g^{i}_{p}) \in \SO(3,\RR)$, which defines a representation of the rotation group $\SO(3,\RR)$ on the vector space $\Ela$, simply written as
\begin{equation*}
  \overline{\bC} = g\cdot \bC.
\end{equation*}

Since we are working on the Euclidean $3$-space, from now on, no distinction will be made between covariant and contravariant indices, and we shall use the notation $C_{ijkl}$.

From a mathematical point of view, different orientations of a given material lead to the following set of elasticity tensors
\begin{equation*}
  \mathcal{O}_{\bC}= \set{g\cdot \bC;\; g \in \SO(3,\RR)},
\end{equation*}
which is called the $\SO(3,\RR)$-\emph{orbit} of $\bC$. Hence, in geometric terms, an elastic material is a point in the \emph{orbit space} $\Ela/\SO(3,\RR)$.

An orbit space has a complicated structure in general. It is not a smooth manifold in general, due to the fact that different orbits may have different symmetry classes (it is however a smooth manifold, when the action is proper and free, that is when each isotropy group is trivial. This is the case for homogeneous spaces, for instance). In the case of the elasticity tensor, it is known that there are eight different symmetry classes~\cite{FV1996}, ranging from complete anisotropy (triclinic materials) to total isotropy. For a given elasticity tensor, the nature of its orbit depends deeply on its symmetry class. For instance, for an isotropic material, we have
\begin{equation*}
  g\cdot \bC = \bC, \qquad \forall g \in \SO(3,\RR).
\end{equation*}
In that case, the orbit of $\bC$ is reduced to $\bC$ itself and the rotation group is thus invisible. However, for any other symmetry class, an elasticity tensor and its orbit should never be confound anymore.

Consider now the measurements of the same (unknown) anisotropic elastic constants in two different labs, and suppose that there is no way to choose, \textit{a priori}, a specific orientation of the material\footnote{This means, in particular, that we do not have any information on the microstructure, or that this information does not allow us to choose a specific orientation.}. Then, the two measurements will result in two different sets of elastic constants. How can one decides whether the two sets of constants describe, or not, the same material? This question was asked by Boehler et al. in~\cite{BKO1994} and can be recast as: \emph{which parameters can be used to characterize intrinsic elastic properties of a given material?}

To answer this question, we need to define $\SO(3,\RR)$-\emph{invariant functions} on $\Ela$ which distinguish different materials. These sets of invariant functions, which \emph{separate the orbits}, are described in the literature under the generic name of a \emph{functional basis}~\cite{WP1964,BKO1994}. There is however no known algorithm to obtain such a set of parameters. This is the reason why we have chosen, following~\cite{BKO1994}, to study this question through \emph{group representation theory} and focus on \emph{polynomial invariants} and the determination of an \emph{integrity basis}, where computations are possible.

\section{Harmonic decomposition}
\label{sec:harmonic-decomposition}

Like a periodic signal can be decomposed into elementary sinusoidal functions, using the Fourier decomposition, any 3D tensor space $V$ can be decomposed into a finite direct sum of spaces which correspond to \emph{irreducible representations} of the rotation group $\SO(3,\RR)$, known as \emph{harmonic tensor spaces} $\HT{n}(\RR^{3})$~\cite{GSS1988,Ste1994}, and defined as follows.

Let $\ST{n}(\RR^{3})$ be the space of \emph{totally symmetric} tensors of order $n$ on $\RR^{3}$ (an $n$-order tensor is understood here as a $n$-linear form on $\RR^{3}$). Contracting two indices $i,j$ on a totally symmetric tensor $\bT$ does not depend on the particular choice of the pair $i,j$. Thus, we can refer to this contraction without any reference to a particular choice of indices. We will denote this contraction as $\tr \bT$, which is a totally symmetric tensor of order $n-2$ and call it the \emph{trace} of $\bT$.

\begin{defn}
  A \emph{harmonic tensor} of order $n$ is a totally symmetric tensor $\bT$ in $\ST{n}(\RR^{3})$ such that $\tr \bT = 0$. The space of harmonic tensors of order $n$ will be denoted by $\HT{n}(\RR^{3})$ (or simply $\HT{n}$, if there is no ambiguity). It is a sub-vector space of $\ST{n}(\RR^{3})$ of dimension $2n+1$.
\end{defn}

The rotation group $\SO(3,\RR)$ acts on $\ST{n}(\RR^{3})$ by the rule
\begin{equation*}
  (g \cdot \bT)(\vv_1,\dotsc,\vv_n) := \bT(g^{-1}\vv_1,\dotsc,g^{-1}\vv_{n}), \qquad g \in \SO(3,\RR).
\end{equation*}
The sub-space $\HT{n}(\RR^{3})$ of $\ST{n}(\RR^{3})$ is invariant under the action of $\SO(3,\RR)$. It is moreover \emph{irreducible} (it has no proper non-trivial invariant sub-space) and every irreducible representation of the rotation group $\SO(3,\RR)$ is isomorphic to some $\HT{n}(\RR^{3})$, see for instance~\cite{GSS1988,Ste1994}. Therefore, every $\SO(3,\RR)$-representation $V$ splits into a direct sum of \emph{harmonic tensor spaces} $\HT{n}(\RR^{3})$.

The space of elasticity tensors admits the following harmonic decomposition which was first obtained by Backus~\cite{Bac1970} (see also~\cite{Bae1993,FV1996,FV2006}):
\begin{equation}\label{eq:elastic-harmonic-decomposition}
  \Ela \simeq 2 \HT{0} \oplus 2 \HT{2} \oplus \HT{4}
\end{equation}
where $\simeq$ indicates an $\SO(3,\RR)$-equivariant isomorphism.
\begin{prop}\label{prop:boehler-decomposition}
  Each $\bC \in \Ela$ can be written as
  \begin{equation}\label{eq:explicit-elasticity-decomposition}
    \bC = (\lambda, \mu, \ba, \bb, \bD),
  \end{equation}
  where $\lambda, \mu \in \HT{0}$, $\ba, \bb \in \HT{2}$ and $\bD \in \HT{4}$.
\end{prop}

\begin{rem}
  It is worth noting that if several identical factors, isomorphic to the same $\HT{n}$, appear in the decomposition of $V$, the explicit isomorphism that realizes this decomposition is not uniquely defined. In the case of the elasticity tensor, any invertible linear combination of $\ba$ and $\bb$ or $\lambda$ and $\mu$ lead to another irreducible decomposition of the elasticity tensor. There is thus an ambiguity in the choice of the two numbers $\lambda,\mu$ and the two deviators $\ba, \bb$. For instance, one can decide that $\lambda,\mu$ are the the \emph{Lam\'{e} numbers}, but one could also use the \emph{shear} $G$ and the \emph{bulk modulus} $K$, which are related to the former by the linear relations
  \begin{equation*}
    G = \mu, \qquad K = \lambda + \frac{2}{3}\mu.
  \end{equation*}
  Concerning the two deviators $\ba,\bb$, one could decide to use the deviatoric part of the \emph{dilatation} tensor $\mathbf{d}$ and the \emph{Voigt} tensor $\mathbf{v}$ (see~\cite{CM1987,Cow1989}), defined respectively as
  \begin{equation*}
    d_{ij} := \sum_{k=1}^{3} C_{kkij}, \quad
    v_{ij} := \sum_{k=1}^{3} C_{kikj}.
  \end{equation*}
  One could also decide to use their following linear combinations as in~\cite{BKO1994}:
  \begin{equation*}
    \ba = \frac{1}{7}( 5\dev \mathbf{d} - 4 \dev \mathbf{v}) ,\quad  \bb = \frac{1}{7} ( -2 \dev \mathbf{d} + 3  \dev \mathbf{v}),
  \end{equation*}
  where $\dev$ indicates the traceless part of a second order symmetric tensor. Nevertheless, all the polynomial invariants formula given in section~\ref{sec:explicit-computations} are independent of these choices.
\end{rem}
Explicit and, due to the aforementioned remark, sometimes different decompositions associated to~\eqref{eq:elastic-harmonic-decomposition} are provided in~\cite{Bac1970,Ona1984,BKO1994,FV1996,FV2006}.

\section{Polynomial invariants}
\label{sec:polynomial-invariants}

Let us now recall some general facts about finite dimensional representations of $\SO(3,\RR)$ that we shall apply to the space $\Ela$. Let $V$ be a finite-dimensional (real) linear representation of $\SO(3,\RR)$. The action will be denoted by $\vv \mapsto g\cdot \vv$, where $g \in \SO(3,\RR)$, and $\vv \in V$. This action can be extended to the algebra $\RR[V]$ of polynomial functions defined on $V$ by the following rule:
\begin{equation*}
  (g\cdot P)(\vv) := P(g^{-1} \cdot \vv)
\end{equation*}
where $P \in \RR[V]$, $\vv \in V$ and $g \in \SO(3,\RR)$. The set of all \emph{invariant polynomials} is a sub-algebra of $\RR[V]$, denoted by $\RR[V]^{\SO(3,\RR)}$ and called the \emph{invariant algebra}. This algebra is moreover \emph{finitely generated}~\cite{Hil1993}. Since moreover, the algebra $\RR[V]$ is the direct sum of spaces of homogeneous polynomials of fixed total degree and that the action of $\SO(3,\RR)$ preserves this graduation, we can always find a generating set of $\RR[V]^{\SO(3,\RR)}$ build up from homogeneous polynomials (see~\cite[Page 227]{GW2009}).

\begin{defn}
  A finite set $\set{J_{1}, \dotsc , J_{N}}$ of invariant homogeneous polynomials on $V$ is called an \emph{integrity basis} if every invariant polynomial on $V$ can be written as a polynomial in $J_{1}, \dotsc , J_{N}$. An integrity basis is said to be \textit{minimal} if none of its elements can be expressed as a polynomial on the others.
\end{defn}

\begin{rem}
  Even if a minimal integrity basis is not unique, all of them have the same cardinality and the list of the degrees of the generators must be the same~\cite{Spe1971,GW2009}.
\end{rem}

\begin{rem}
  It is also worth noting that this definition does not exclude that some polynomial relations, known as \emph{syzygies}, may exist between generators. Such relations can not be avoided in most cases and their determination is a difficult problem \cite{Shi1967,AKP2014a}.
\end{rem}

An important property of polynomial invariants for a real representation of a compact group (and thus of any integrity basis), attributed to Weyl~\cite{Wey1997} (see also~\cite[Appendix C]{AS1983}), is that they \emph{separate the orbits}, which means that:
\begin{equation*}
  P(\vv_{1}) = P(\vv_{2}),\qquad \forall P \in \RR[V]^{\SO(3,\RR)},
\end{equation*}
if and only if $\vv_{1} = g\cdot \vv_{2}$ for some $g\in \SO(3,\RR)$.

\begin{ex}
  For instance, two vectors $\vv,\ww \in \RR^{3}$ are rotates of each other if and only if they have the same norms $\norm{\vv} = \norm{\ww}$. In fact, the invariant algebra in this case is generated by the single homogeneous polynomial $P(\vv) := \norm{\vv}^{2}$.
\end{ex}

We could omit the condition of polynomiality and obtain a more general definition of an invariant function on $V$; for example, one may consider rational, smooth, continuous, \ldots , invariant functions. In this general framework, we are lead to the following definition.

\begin{defn}
  A finite set $\mathcal{F}=\set{s_{1},\dotsc,s_{n}}$ of invariant functions on $V$ is called a \emph{functional basis} if $\mathcal{F}$ separates the orbits. A functional basis is \emph{minimal} if no proper subset of it is a functional basis.
\end{defn}

\begin{rem}
  Note that a polynomial functional basis may not be an integrity basis. Consider, for instance, the space, of symmetric, traceless second-order tensors on $\RR^{3}$. An integrity basis, for the action of $\SO(3,\RR)$, is constituted by the polynomial invariants $I_{2}(\ba) = \tr(\ba^{2})$, and $I_{3}(\ba) = \tr(\ba^{3})$. Since $\tr(\ba^{2})>0$, the set $\set{I_{2}^{2}, I_{3}}$ is a functional basis but it is not an integrity basis.
\end{rem}

\begin{rem}
  Contrary to an integrity basis, there is no reason that two polynomial minimal functional bases have the same cardinal number. Moreover, there is no known algorithm to determine the minimum cardinal number of a polynomial separating set. It is not even easy to check if a given functional basis is minimal or not.
\end{rem}

Many results are known on invariants for an arbitrary number of vectors, skew and symmetric second-order tensors~\cite{Liu1982,Boe1987,Zhe1994}. Some of them concerns integrity bases~\cite{Riv1955,SR1958/1959,SR1962}, others functional bases~\cite{Smi1971,Wan1970}. According to them, it is possible to find \emph{polynomial functional basis} with a smaller cardinality than that of a \emph{minimal integrity basis}. For instance, the cardinality of a minimal integrity basis for the action of $\SO(3,\RR)$ on the direct sum of 3 second-order symmetric tensors is 28, but there exists a functional basis (which do not generate the invariant algebra) consisting of 22 polynomial invariants~\cite{Boe1977}. However, no general algorithm currently exists to produce a minimal functional basis, whereas there are algorithms to compute a minimal \emph{integrity basis}. For higher-order tensors, results are rather un-complete and restricted to particular cases. The reason lies in the fact that classical geometrical methods used for low-order tensors cease to work for tensors of order $\ge 3$. Even if not directly formulated in these terms, this point seems to have been clear to some authors in this field~\cite{BKO1994,Smi1994,SB1997}.

Getting back to the elasticity tensor, the harmonic decomposition~\eqref{eq:explicit-elasticity-decomposition}, allows to consider an invariant function of $\bC$ as a function of the variables $(\lambda, \mu, \ba, \bb, \bD)$. An invariant which depends only on one of the variables is called a \emph{simple invariant} and an invariant which depends on two or more of them is called a \emph{joint invariant}. For instance, the scalars $\lambda, \mu$ are simple invariants, $\tr(\ba^{2})$ and $\tr(\ba^{3})$ are simple invariants which generate the invariant algebra of $\HT{2}$, and similarly for $\bb$. In~\cite{BKO1994}, Boehler, Onat and Kirillov exhibited for the first time, using previous calculations by Shioda~\cite{Shi1967} and von Gall~\cite{vGal1880}, nine simple invariants of $\bD$ which generate the invariant algebra of $\HT{4}$.

\begin{prop}\label{prop:inv-H4}
  Let $\bD \in \HT{4}$ and set:
  \begin{equation*}
    \begin{array} {lll}
      \mathbf{d}_{2} := \tr_{13} \bD^{2},                            & \mathbf{d}_{3} := \tr_{13} \bD^{3},                            & \mathbf{d}_{4} := \mathbf{d}_{2}^{2},                               \\
      \mathbf{d}_{5} := \mathbf{d}_{2}(\bD \mathbf{d}_{2}),          & \mathbf{d}_{6} := \mathbf{d}_{2}^{3},                          & \mathbf{d}_{7} := \mathbf{d}_{2}^{2} (\bD \mathbf{d}_{2})           \\
      \mathbf{d}_{8} := \mathbf{d}_{2}^{2}( \bD^{2} \mathbf{d}_{2}), & \mathbf{d}_{9} := \mathbf{d}_{2}^{2}( \bD \mathbf{d}_{2}^{2}), & \mathbf{d}_{10} := \mathbf{d}_{2}^{2} (\bD^{2} \mathbf{d}_{2}^{2}).
    \end{array}
  \end{equation*}
  An integrity basis of $\HT{4}$ is given by the nine fundamental invariants:
  \begin{equation*}
    J_{k} := \tr \mathbf{d}_{k} , \qquad k = 2, \dotsc ,10.
  \end{equation*}
\end{prop}

\begin{rem}
  The first 6 invariants $J_{2}, \dotsc , J_{7}$ are algebraically independent, whereas the last 3 ones $J_{8},J_{9},J_{10}$ are linked to the formers by polynomial relations. These \emph{relations} were computed in~\cite{Shi1967}.
\end{rem}

To obtain an integrity basis of the elasticity tensor, it is necessary to complete these results by including \emph{joint invariants} of $\ba, \bb, \bD$. For instance, a minimal integrity basis for $\HT{2}\oplus\HT{2}$ is known~\cite{Smi1965,You1898}.

\begin{prop}\label{prop:inv-H2H2}
  An integrity basis of $\HT{2}\oplus\HT{2}$ is given by the eight fundamental invariants:
  \begin{align*}
    I_2 & :=\tr(\ba^{2}), \quad I_3 :=\tr(\ba^{3}), \quad J_2:=\tr(\bb^{2}), \quad J_3:=\tr(\bb^{3})              \\
    K_2 & :=\tr(\ba\bb), \quad K_3:=\tr(\ba^{2}\bb), \quad L_3:=\tr(\ba\bb^{2}), \quad K_4 :=\tr(\ba^{2}\bb^{2}). \\
  \end{align*}
\end{prop}

In \cite{BKO1994}, the authors tried to compute all joint invariants but realized that running classical algorithms by hand would be \emph{prohibitively long}. They nevertheless formulate a \emph{generic} hypothesis on  $\bD$ which results in a \emph{weak} functional basis constituted by 39 polynomial invariants able to separate \emph{generic tensors}. As pointed by the authors themselves this hypothesis, which only concerns a subset of triclinic materials, is not satisfactory. In the present work, the combination of non-trivial tools from classical invariant theory (described in the next sections) with the use of a \emph{Computer Algebra System} (CAS) software allows us to conduct the complete computation leading to the following result.

\begin{thm}\label{thm:main-result}
  The polynomial invariant algebra of $\Ela$ is generated by a minimal basis of 297 homogeneous invariant polynomials, resumed
  in table~\ref{tab:Inv844}, which describes the number and the total degree of simple and joint invariants of this basis.
\end{thm}

\begin{table}[h]
  \setlength{\arraycolsep}{6pt}
  \begin{equation*}
    \begin{array}{c||cccccc|c}
      \text{degree} & \HT{4} & \HT{2}     & \HT{0}     & \HT{2}\oplus \HT{2} & \HT{4}\oplus \HT{2} & \HT{4}\oplus \HT{2}\oplus \HT{2} & \Sigma \\\hline\hline
      1             & -      & -          & 1          & -                   & -                   & -                                & 2      \\
      2             & 1      & 1          & -          & 1                   & -                   & -                                & 4      \\
      3             & 1      & 1          & -          & 2                   & 2                   & 1                                & 10     \\
      4             & 1      & -          & -          & 1                   & 4                   & 6                                & 16     \\
      5             & 1      & -          & -          & -                   & 7                   & 18                               & 33     \\
      6             & 1      & -          & -          & -                   & 10                  & 36                               & 57     \\
      7             & 1      & -          & -          & -                   & 11                  & 53                               & 76     \\
      8             & 1      & -          & -          & -                   & 10                  & 45                               & 66     \\
      9             & 1      & -          & -          & -                   & 5                   & 10                               & 21     \\
      10            & 1      & -          & -          & -                   & 2                   & 2                                & 7      \\
      11            & -      & -          & -          & -                   & 1                   & 3                                & 5      \\ \hline
      Tot           & 9      & 2 \times 2 & 2 \times 1 & 4                   & 2 \times 52         & 174                              & 297    \\
    \end{array}
  \end{equation*}
  \caption{Minimal integrity basis for the elasticity tensor.}
  \label{tab:Inv844}
\end{table}

It can be observed that the number of elementary invariants of each degree provided by our theorem confirms some previously published results \cite{BH1995,Ahm2002,Nor2007,Iti2015}.

\section{Complexification of the problem}
\label{sec:harmonic-tensor-versus-binary-forms}

Before entering the details of computations for the invariants in 3D, let us recall first the situation in 2D. To compute an integrity basis for a real representation $V$ of the rotation group $\SO(2,\RR)$, $V$ is first split into irreducible representations \cite{AKO2016}
\begin{equation*}
  V_{1} \oplus \dotsb \oplus V_{r}.
\end{equation*}
It is also useful to complexify the problem, which means extending the representation to the complexified space $V^{\CC} := V \oplus iV$.
Now each irreducible complex representation of $\SO(2,\RR)$ is one-dimensional, indexed by $n \in \ZZ$, and represented by
\begin{equation*}
  \rho_{n}(\theta) \cdot z := e^{in\theta}z,
\end{equation*}
where $\theta\in \SO(2,\RR)$ and $z \in \CC$. Let $\CC_{n}$ denote the representation $(\CC,\rho_{n})$. Then, for each real representation $V$ of $\SO(2,\RR)$, the complexified space $V^{\CC}$ is isomorphic to
\begin{equation*}
  \CC_{m_{1}} \oplus \dotsb \oplus \CC_{m_{r}} \oplus \CC_{-m_{1}} \oplus \dotsb \oplus \CC_{-m_{r}}.
\end{equation*}
The monomials
\begin{equation*}
  z_{1}^{\alpha_{1}} \dotsm z_{r}^{\alpha_{r}} \bar{z}_{1}^{\beta_{1}} \dotsm \bar{z}_{r}^{\beta_{r}}
\end{equation*}
span stable one-dimensional subspaces of $\CC[V^{\CC}]$ and the invariant algebra of $V^{\CC}$ is generated by the monomials which satisfy the \emph{Diophantine equation}
\begin{equation}\label{eq:diophantine-equation}
  m_{1}\alpha_{1} + \dotsb + m_{r}\alpha_{r} - m_{1}\beta_{1} - \dotsb - m_{r}\beta_{r} = 0,
\end{equation}
where $(\boldsymbol{\alpha},\boldsymbol{\beta}) := (\alpha_{1},\dotsb,\alpha_{r},\beta_{1},\dotsb, \beta_{r}) \in \NN^{2r}$. A solution $(\boldsymbol{\alpha},\boldsymbol{\beta})$ is called \emph{irreducible} if it is not the sum of two non--trivial solutions. It is, by the way, a classical result \cite{Stu2008} that there is only a finite number of irreducible solutions of~\eqref{eq:diophantine-equation}. Moreover, there exists algorithms to compute them \cite{BI2010}. Thus, an integrity basis of the invariant algebra of $V^{\CC}$ is given by monomials corresponding to irreducible solutions of the Diophantine equation~\eqref{eq:diophantine-equation}. Following a work of Pierce \cite{Pie1995}, this approach was applied to plane elasticity by Vianello \cite{Via1997,FV2014} and in a related way by Verchery some years before \cite{Ver1982}.

In 3D, the scheme is more or less similar but the complexification process is much more sophisticated. Complex irreducible representations of $\SO(3,\RR)$ are no more \emph{one-dimensional} and the description of polynomial invariants requires additional tools. First, it is preferable to use the space $\Hn{n}(\RR^{3})$ (of harmonic polynomials on $\RR^{3}$) as a model for irreducible representations, rather than the space $\HT{n}(\RR^{3})$ of harmonic tensors (see Appendix~\ref{sec:harmonic-polynomials}). Then, each irreducible representation $\Hn{n}(\RR^{3})$ of the real group $\SO(3,\RR)$ can be complexified to obtain an irreducible representation $\Hn{n}(\CC^{3})$ of the complex group $\SO(3,\CC)$, where $\Hn{n}(\CC^{3})$ is the space of complex harmonic, homogeneous polynomials in three variables of degree $n$. This space is closely related to the space of \emph{binary forms} (\textit{i.e.} homogeneous complex polynomials in two variables) of degree $2n$. The object of this section is to describe explicitly this relationship, which is obtained using the \emph{Cartan map} (see Appendix~\ref{sec:Cartan-map}). Albeit being rather confidential in the field of continuum mechanics, this approach has been explored in some publications \cite{Bac1970,Bae1993,BBS2004b}.

Using the universal cover $\pi: \SL(2,\CC) \to \SO(3,\CC)$, described in Appendix~\ref{sec:Cartan-map},  the $\SO(3,\CC)$-representation $\Hn{n}(\CC^{3})$ can be extended into an $\SL(2,\CC)$ representation, writing
\begin{equation*}
  \gamma \cdot \bh := \pi(\gamma) \cdot \bh, \qquad \gamma \in \SL(2,\CC),\quad \bh \in \Hn{n}(\CC^{3}),
\end{equation*}
which remains irreducible. Finite-dimensional irreducible representations of $\SL(2,\CC)$ are all known~\cite{Ste1994}. They correspond to the spaces $\Sn{p}$ of binary forms of degree $p$
\begin{equation*}
  \ff(u,v) := \sum_{k=0}^{p} \binom{p}{k}a_{k}u^{k}v^{p-k},
\end{equation*}
where the action of $\SL(2,\CC)$ is defined as
\begin{equation*}
  (\gamma \cdot \ff) (\bxi) : = \ff(\gamma^{-1} \cdot \bxi), \qquad \gamma \in \SL(2,\CC), \quad \bxi \in \CC^{2},
\end{equation*}
and $\gamma \cdot \bxi$ is the standard action of $\SL(2,\CC)$ on $\CC^{2}$. For dimensional reason, there must exist an equivariant isomorphism between $\Hn{n}(\CC^{3})$ and $\Sn{2n}$ which by the Schur's lemma, is unique up to a multiplicative factor. Such an isomorphism is provided explicitly using the \emph{Cartan map}:
\begin{equation*}
  \phi : \CC^{2} \to \CC^{3}, \qquad (u,v) \mapsto \left(\frac{u^{2} - v^{2}}{2},\frac{u^{2} + v^{2}}{2i},uv\right).
\end{equation*}
The geometric meaning of this mapping and its properties are detailed in Appendix~\ref{sec:Cartan-map}.

\begin{thm}\label{thm:explicit-isomorphism}
  The linear mapping
  \begin{equation*}
    \phi^{*} : \Hn{n}(\CC^{3}) \to \Sn{2n},\qquad \bh \mapsto \phi^{*} \bh := \bh \circ \phi,
  \end{equation*}
  where
  \begin{equation*}
    (\phi^{*} \bh)(u,v) = \bh \left( \frac{u^{2}-v^{2}}{2}, \frac{u^{2}+v^{2}}{2i}, uv \right)
  \end{equation*}
  is an $\SL(2,\CC)$-equivariant isomorphism.
\end{thm}

\begin{proof}
  Since $\Hn{n}(\CC^{3})$ and $\Sn{2n}$ have the same complex dimension $2n+1$, it is sufficient to prove that the linear mapping $\phi^{*}$ is surjective. Let
  \begin{equation*}
    \ff(u,v) := \sum_{k=0}^{2n} \binom{2n}{k}a_{2n-k}u^{2n-k}v^{k}\in \Sn{2n}.
  \end{equation*}
  For each $k$ make the substitution
  \begin{equation*}
    u^{2n-k}v^{k} \rightarrow
    \left\{ \begin{array}{cc}
    z^{k}(x+iy)^{n-k}, & \text{if } 0 \le k \le n\\
    z^{2n-k}(-x+iy)^{k-n}, & \text{if } n \le k \le 2n
    \end{array}
    \right.
  \end{equation*}
  We obtain this way a homogeneous polynomial $\bp$ in three variables and of degree $n$ such that $\phi^{*}(\bp) = \ff$. Now, let $\bh_{0}$ be the harmonic component of $\bp$ in the harmonic decomposition:
  \begin{equation*}
    \bp = \bh_{0} + \qq\bh_{1} + \dotsb + \qq^{r}\bh_{r},
  \end{equation*}
  as detailed in Appendix~\ref{sec:harmonic-polynomials}, where $\qq := x^{2} + y^{2} + z^{2}$. We get thus
  \begin{equation*}
    \ff(u,v) = \bp \left(\frac{u^{2}-v^{2}}{2},\frac{u^{2}+v^{2}}{2i},uv\right) = \bh_{0} \left(\frac{u^{2}-v^{2}}{2},\frac{u^{2}+v^{2}}{2i},uv\right),
  \end{equation*}
  because $\qq$ vanishes on the isotropic cone
  \begin{equation*}
    C := \set{(x,y,z) \in \CC^{3};\; x^{2} + y^{2} + z^{2} = 0}.
  \end{equation*}
  The $\SL(2,\CC)$-equivariance is a direct consequence of Lemma~\ref{lem:Cartan-equiv}, which achieves the proof.
\end{proof}

\begin{ex}
  A binary form of degree 4
  \begin{equation*}
    \ff(u,v)=a_{0}u^{4}+4a_{1}u^{3}v+6a_{2}u^{2}v^{2}+4a_{3}uv^{3}+a_{4}v^{4}\in \Sn{4}
  \end{equation*}
  corresponds to the harmonic polynomial
  \begin{align*}
    \bh(x,y,z) & = \left(a_{0} + a_{4} - 2a_{2}\right) x^{2} - \left(a_{0}+a_{4} + 2a_{2}\right)y^{2} + 4a_{2}z^{2}
    \\
               & + 2i(a_{4} - a_{0})xy + 4(a_{3}-a_{1})xz + 4i(a_{1} + a_{3})yz.
  \end{align*}
\end{ex}

\begin{ex}
  A binary form of degree 8
  \begin{equation*}
    \bg(u,v) = \sum_{k=0}^8 \binom{8}{k} b_{k}u^{8-k}v^{k}\in \Sn{8}
  \end{equation*}
  corresponds to the harmonic polynomial
  \begin{align*}
    \bh(x,y,z) & = \left( 6\,b_{{4}}-4\,b_{{2}}-4\,b_{{6}}+b_{{0}}+b_{{8}} \right) {x}^{4}
    +\left( b_{{0}}+6\,b_{{4}}+4\,b_{{2}}+4\,b_{{6}}+b_{{8}} \right) {y}^{4}
    +16\,b_{{4}}{z}^{4}															\\
               & + 4\left( -2\,ib_{{2}}+\,ib_{{0}}+2\,ib_{{6}}-\,ib_{{8}} \right) {x}^{3}y
    + 8\left( 3\,b_{{5}}+\,b_{{1}}-\,b_{{7}}-3\,b_{{3}} \right) {x}^{3}z		\\
               & - 8\left( \,ib_{{7}}+3\,ib_{{3}}+\,ib_{{1}}+3\,ib_{{5}} \right) {y}^{3}z
    + 4\left( 2\,ib_{{6}}-\,ib_{{0}}+\,ib_{{8}}-2\,ib_{{2}} \right) x{y}^{3}	\\
               & + 32\left( \,b_{{3}}-\,b_{{5}} \right) x{z}^{3}
    + 32i\left( \,b_{{3}}+\,b_{{5}} \right) y{z}^{3} 							\\
               & + 6\left( -\,b_{{0}}-\,b_{{8}}+2\,b_{{4}} \right) {x}^{2}{y}^{2}
    + 24\left( \,b_{{2}}-2\,b_{{4}}+\,b_{{6}} \right) {x}^{2}{z}^{2} 		\\
               & - 24\left( \,b_{{2}}+\,b_{{6}}+2\,b_{{4}} \right) {y}^{2}{z}^{2}
    + 48i\left(-\,b_{{6}}+\,b_{{2}} \right) xy{z}^{2}							\\
               & + 24\left( -\,b_{{1}}+\,b_{{7}}-\,b_{{3}}+\,b_{{5}} \right) x{y}^{2}z
    + 24i\left( \,b_{{7}}-\,b_{{3}}-\,b_{{5}}+\,b_{{1}} \right) {x}^{2}yz.
  \end{align*}
\end{ex}

\begin{rem}\label{rem:Real_Binary_Forms}
  Binary forms
  \begin{equation*}
    \ff(u,v) := \sum_{k=0}^{n} \binom{2n}{k}a_{2n-k}u^{2n-k}v^{k}
  \end{equation*}
  in $\Sn{2n}$ which are images by $\phi^{*}$ of \emph{real} harmonic polynomials in $\Hn{n}(\RR^{3})$ are defined by the following linear equations:
  \begin{equation}\label{eq:real-binary-forms}
    a_{2n-k} = (-1)^{n-k} \overline{a_{k}}, \qquad 0 \le k \le n.
  \end{equation}
  They can also be characterized by the following equivalent condition
  \begin{equation}\label{eq:real-binary-forms-functional}
    \bar{\ff}(-v,u) = (-1)^{n} \ff(u,v).
  \end{equation}
  These binary forms generate a \emph{real} vectorial subspace of $\Sn{2n}$, invariant by $\SU(2)$.
\end{rem}

\section{Invariants and covariants of binary forms}
\label{sec:covariants}

The method that have been used to compute the invariants of the elasticity tensor is known as \emph{Gordan's algorithm}. A detailed description of it can be found in~\cite{Oli2016}. This algorithm is based on an extension of the notion of invariants called \emph{covariants}, which is the subject of this section.

\subsection{Covariants of a binary form}

\begin{defn}
  Let $\ff \in \Sn{n}$ be a binary form. A \emph{covariant} of the binary form $\ff$ is a polynomial
  \begin{equation*}
    \hh(\ff,\bxi) = \sum_{i,j} h_{ij}(\ff)u^{i}v^{j},
  \end{equation*}
  where each $h_{ij}(\ff)$ are polynomials in the coefficients $\ff=(a_k)$ and such that
  \begin{equation}\label{eq:def-covariant}
    \hh(\gamma \cdot \ff, \bxi) = \hh(\ff,\gamma^{-1} \cdot \bxi).
  \end{equation}
  The set of covariants of a binary form $\ff$ is a \emph{sub-algebra} of $\CC[a_{1}, \dotsc , a_{n},u,v]$, called the \emph{covariant algebra} of $\Sn{n}$ and noted $\cov{\Sn{n}}$.
\end{defn}

\begin{rem}
  Note that equation~\eqref{eq:def-covariant} can be recast as
  \begin{equation*}
    \hh(\gamma \cdot \ff, \gamma\cdot\bxi) = \hh(\ff,\bxi),
  \end{equation*}
  and a covariant can also be thought as a polynomial invariant of $\Sn{n}\oplus \CC^{2}$. We have therefore
  \begin{equation*}
    \cov{\Sn{n}} = \CC[\Sn{n}\oplus \CC^{2}]^{\SL(2,\CC)}.
  \end{equation*}
\end{rem}

\begin{rem}
  Given a covariant $\hh(\ff,\bxi)$, the total degree in the variables $a_{k}$ is called the \emph{degree} of $\hh$ whereas the total degree in the variables $u,v$ is called the \emph{order} of $\hh$. The sub-algebra of covariants of order $0$ in $\cov{\Sn{n}}$ is the invariant algebra of $\Sn{n}$, noted also $\inv{\Sn{n}}$.
\end{rem}

\begin{ex}
  Let
  \begin{equation*}
    \ff(\bxi) := a_{0}u^{3} + 3a_{1}u^{2}v + 3a_{2}uv^{2} + a_{3}v^{3},
  \end{equation*}
  be a binary form of degree 3. Its \emph{Hessian}
  \begin{align*}
    \hh(\ff,\bxi) & :=	\frac{\partial^{2} \ff}{\partial u^{2}}\frac{\partial^{2} \ff}{\partial v^{2}}-\left( \frac{\partial^{2} \ff}{\partial u\partial v}\right)^{2} \\
                  & = 36(a_{0}a_{2} - a_{1}^{2})u^{2} + 36(a_{0}a_{3} - a_{1}a_{2})uv + 36(a_{1}a_{3} - a_{2}^{2})v^{2},
  \end{align*}	
  is a covariant of $\ff$ of order $2$ and degree $2$.
\end{ex}

\begin{rem}
  The notion of covariant can of course be extended to several binary forms $\ff_{1}, \dotsc , \ff_{p}$, in which case the coefficients $h_{ij}$ of the covariant are polynomials in all the coefficients of the $\ff_{i}$'s.
\end{rem}

A way to generate covariants is to use an $\SL(2,\CC)$-equivariant bi-differential operator, called the \emph{Cayley operator} and defined by
\begin{equation*}
  \Omega_{\alpha\beta} := \frac{\partial^{2}}{\partial u_{\alpha}\partial v_{\beta}} - \frac{\partial^{2}}{\partial v_{\alpha}\partial u_{\beta}}.
\end{equation*}

\begin{defn}
  The \emph{transvectant} of index $r$ of two binary forms $\ff\in \Sn{n}$ and $\bg\in \Sn{p}$, noted $\trans{\ff}{\bg}{r}$, is defined as the following binary form
  \begin{equation*}
    \trans{\ff}{\bg}{r}(\bxi) := \left\{\Omega^{r}_{\alpha\beta}(\ff(\bxi_{\alpha}) \bg(\bxi_{\beta}))\right\}_{\bxi_{\alpha} = \bxi_{\beta} = \bxi},
  \end{equation*}
  which is of order $n+p-2r$ (for $r\leq \min(n,p)$, it is zero otherwise), where $\Omega^{r}_{\alpha\beta}$ is the $r$-th iterate of the operator $\Omega_{\alpha\beta}$. It is also given by the explicit formula:
  \begin{equation}\label{eq:Trans_Explicit}
    \trans{\ff}{\bg}{r} = \sum_{i=0}^{r}(-1)^{i} \binom{r}{i} \frac{\partial^{r} \mathbf{f}}{\partial^{r-i}u \partial^{i} v}
    \frac{\partial^{r} \mathbf{g}}{\partial^{i}u \partial^{r-i}v}.
  \end{equation}
\end{defn}

\begin{rem}
  Transvectants are connected with the famous \emph{Clebsch--Gordan formula}:
  \begin{equation*}
    \Sn{n} \otimes \Sn{p} \simeq \bigoplus_{r=0}^{\min(n,p)} \Sn{n+p-2r},
  \end{equation*}
  which describes how the tensor product of two $\SL(2,\CC)$-irreducible representations splits into irreducible factors (see for instance~\cite{Ste1994}). The transvectant $\trans{\ff}{\bg}{r}$ corresponds to an explicit projection
  \begin{equation*}
    \Sn{n} \otimes \Sn{p} \to \Sn{n+p-2r},
  \end{equation*}
  which is, up to a scaling factor, unique by Schur's lemma.
\end{rem}

The key point is that by iterating the process of taking transvectants
\begin{equation*}
  \ff_{1}, \dotsc ,\ff_{p} \quad \trans{\ff_{i}}{\ff_{j}}{r}, \quad \trans{\ff_{i}}{\trans{\ff_{j}}{\ff_{k}}{r}}{s}, \quad \dotsc
\end{equation*}
that we shall call \emph{iterated transvectants}, one generates the full algebra $\cov{V}$, where $V = \Sn{n_{1}}\oplus \dotsc \oplus \Sn{n_{p}}$. Restricting to covariants of order $0$, the invariant algebra $\inv{V}$ is also generated. This important fact is summarized in the following theorem (see~\cite{GY2010,Olv1999} for details).

\begin{thm}
  Let $V = \Sn{n_{1}}\oplus \dotsc \oplus \Sn{n_{p}}$. Then, the algebras $\cov{V}$ and $\inv{V}$ are generated by iterated transvectants.
\end{thm}

Iterated transvectants are thus an infinite system of generators for the invariant and the covariant algebras. The main goal of nineteenth century's invariant theory~\cite{Gor1868,GY2010,Olv1999} was to prove moreover that $\cov{V}$ and $\inv{V}$ were \emph{finitely generated} and to \emph{compute explicitly} minimal integrity bases for these algebras. This goal was achieved first by Gordan~\cite{Gor1868} in 1868 and then by Hilbert~\cite{Hil1993} in 1890 (in a more general setting). The remarkable achievement of Gordan was that his proof was constructive (and extremely efficient). It is now known as the \emph{Gordan algorithm} (see~\cite{Oli2016}) and will be shortly reviewed in the next section.

\subsection{Gordan's algorithm}
\label{subsec:Gordan-algorithm}

There are two versions of \emph{Gordan's algorithm}. One of them is devoted to the calculation of an integrity basis for the covariant algebra of a \emph{single binary form}. It produces a basis for $\cov{\Sn{n}}$, already knowing bases for $\cov{\Sn{k}}$, for each $k < n$. The second version is devoted to the calculation of an integrity basis for the covariant algebra of several binary forms. It \emph{produces a basis} for $\cov{V_{1}\oplus V_{2}}$, already knowing bases for $\cov{V_{1}}$ and $\cov{V_{2}}$, where $V_{1},V_{2}$ are direct sums of some $\Sn{k}$. Both of them rely on the resolution of a Diophantine equation such as \eqref{eq:diophantine-equation}. It is the second version that has been used to produce the tables of section~\ref{sec:explicit-computations} and that we shortly outline next (a more detailed treatment of these algorithms is provided in~\cite{Oli2016}).

Let $\ff_{1},\dotsb,\ff_{p}$ (resp. $\bg_{1},\dotsb,\bg_{q}$) be a finite generating set for $\cov{V_{1}}$ (resp. $\cov{V_{2}}$). The first observation which proof can be found in~\cite{GY2010,Oli2016} is the following result.

\begin{thm}
  The covariant algebra $\cov{V_{1}\oplus V_{2}}$ is generated by transvectants
  \begin{equation}\label{eq:Trans_Joints}
    \trans{\mathbf{f}_{1}^{\alpha_{1}}\cdots\mathbf{f}_{p}^{\alpha_{p}}}{\mathbf{g}_{1}^{\beta_{1}}\cdots\mathbf{g}_{q}^{\beta_{q}}}{r},
  \end{equation}
  where $\alpha_{i},\beta_{i}\in \NN$.
\end{thm}

Now, since $\trans{\ff}{\bg}{r}$ vanishes unless the order of $\ff$ and $\bg$ are $\ge r$, we only have to consider transvectants in \eqref{eq:Trans_Joints} such that:
\begin{equation*}
  \alpha_{1}a_{1}+\cdots+\alpha_{p}a_{p} \ge r, \qquad \beta_{1}b_{1}+\cdots+\beta_{q}b_{q} \ge r,
\end{equation*}
where $a_{i}$ is the order of $\ff_{i}$ and $b_{j}$ is the order of $\bg_{j}$. Thus any non-vanishing transvectant
\begin{equation*}
  \boldsymbol{\tau}=\trans{\mathbf{f}_{1}^{\alpha_{1}}\cdots\mathbf{f}_{p}^{\alpha_{p}}}{\mathbf{g}_{1}^{\beta_{1}}\cdots\mathbf{g}_{q}^{\beta_{q}}}{r}
\end{equation*}
corresponds to a solution
\begin{equation*}
  \boldsymbol{\kappa}=(\alpha_{1},\dotsb,\alpha_{p},\beta_{1},\dotsb,\beta_{q},u,v,r) \in \NN^{p+q+3}
\end{equation*}
of the Diophantine system
\begin{equation*}
  (S)\: : \: \begin{cases}
  \alpha_{1}a_{1}+\cdots+\alpha_{p}a_{p}&=u+r \\
  \beta_{1}b_{1}+\cdots+\beta_{q}b_{q}&=v+r	
  \end{cases}.
\end{equation*}	
But the \emph{linear Diophantine system} $(S)$ possesses only a finite number of irreducible solutions (which can not be written as a sum of non--trivial solutions) and the result below (see~\cite{Oli2016} for a proof) shows that these irreducible solutions generate $\cov{V_{1}\oplus V_{2}}$.

\begin{thm}[Gordan-1868]\label{thm:Gord_Alg}
  Let $\boldsymbol{\kappa}_{1},\cdots,\boldsymbol{\kappa}_{l}$ be the irreducible solutions of the Diophantine system $(S)$ and let $\boldsymbol{\tau}_{1},\cdots,\boldsymbol{\tau}_{l}$ be the associated transvectants. Then $\cov{V_{1}\oplus V_{2}}$ is generated by $\boldsymbol{\tau}_{1},\cdots,\boldsymbol{\tau}_{l}$.
\end{thm}

\begin{rem}
  Note that the integrity basis $\set{\boldsymbol{\tau}_{1},\cdots,\boldsymbol{\tau}_{l}}$ may not be minimal. Additional reductions on the set $\set{\boldsymbol{\tau}_{1},\cdots,\boldsymbol{\tau}_{l}}$ may be required to produce a minimal basis \cite{Oli2016,Oli2014,OL2015}.
\end{rem}

Since the computation of \emph{joint invariants} requires the knowledge of \emph{simple covariants}, it might be worth to recall what is known about them. Minimal integrity bases for invariant and covariant algebras of $\Sn{2}$, $\Sn{3}$, $\Sn{4}$ were already available since the middle of the nineteenth century~\cite{Cay1861,Gor1868,GY2010}.

\begin{ex}\label{ex:S2}
  The invariant algebra $\inv{\Sn{2}}$ is generated by the discriminant $\triangle = \trans{\ff}{\ff}{2}$. The covariant algebra $\cov{\Sn{2}}$ is generated by $\triangle$ and $\ff$.
\end{ex}

\begin{ex}\label{ex:S4}
  Let $\ff \in \Sn{4}$, and set
  \begin{equation*}
    \mathbf{h} := \trans{\ff}{\ff}{2}, \quad \kk := \trans{\ff}{\mathbf{h}}{1}, \quad \mathbf{i} := \trans{\ff}{\ff}{4}, \quad \mathbf{j} := \trans{\ff}{\mathbf{h}}{4}.
  \end{equation*}
  Then, we have
  \begin{equation*}
    \inv{\Sn{4}}=\CC[\mathbf{i},\mathbf{j}] \quad \text{and} \quad \cov{\Sn{4}} = \CC[\mathbf{i},\mathbf{j},\ff,\mathbf{h},\kk].
  \end{equation*}
\end{ex}

Gordan and his followers~\cite{Gor1868,vGal1874,vGal1880,vGal1888} were able to produce (without the help of a computer) generating sets for invariant/covariant algebras of $\Sn{5}$, $\Sn{6}$, $\Sn{7}$ and $\Sn{8}$. Some of these generating sets were not minimal, and some contained a few errors, but still, this remains a \emph{tour de force}! These results have since been checked and corrected~\cite{Cro2002,Bed2008,Bed2009}. Minimal integrity bases have been computed recently for the invariant algebra of $\Sn{9}$ and $\Sn{10}$ \cite{BP2010,BP2010a} and also for their covariant algebra~\cite{OL2015}. For higher orders, results are conjectural or unknown. An overview of all these results is available in~\cite{Bro2015}.

\subsection{Integrity bases for real tensor spaces}

Once a minimal integrity basis $\set{\boldsymbol{\tau}_{1},\dotsc , \boldsymbol{\tau}_{N}}$ has been provided for the invariant algebra of a space of even degree binary forms
\begin{equation*}
  V := \Sn{2n_{1}}\oplus \dotsc \oplus \Sn{2n_{p}},
\end{equation*}
the question arises how to deduce a minimal integrity basis for the corresponding \emph{real} $\SO(3,\RR)$-representation
\begin{equation*}
  W := \Hn{n_{1}}(\RR^{3})\oplus \dotsc \oplus \Hn{n_{p}}(\RR^{3}).
\end{equation*}

Recall first that the complex spaces $V$ and
\begin{equation*}
  W^{\CC} = \Hn{n_{1}}(\CC^{3})\oplus \dotsc \oplus \Hn{n_{p}}(\CC^{3})
\end{equation*}
are isomorphic $\SL(2,\CC)$-representations (see section~\ref{sec:harmonic-tensor-versus-binary-forms}). Therefore, if we set $J_{k} := \boldsymbol{\tau}_{k} \circ \phi^{*}$, where $\phi^{*}$ is the linear isomorphism introduced in Theorem~\ref{thm:explicit-isomorphism}, the set $\set{J_{1},\dotsc , J_{N}}$ is a minimal integrity basis for the invariant algebra of $W^{\CC}$ as an $\SL(2,\CC)$-representation, and also as an $\SO(3,\CC)$-representation.

\textit{A priori}, each $J_{k}$ belongs to $\CC[W^{\CC}]$. The fundamental observation, now, is that the space of binary forms which correspond to real harmonic polynomials (see remark \ref{rem:Real_Binary_Forms}) is stable under the \emph{transvectant process}. More precisely, if $\ff\in \Sn{2n}$ and $\bg\in \Sn{2p}$ are such that
\begin{equation*}
  \ff(u,v)=\overline{\ff}(-v,u),\quad \bg(u,v)=\overline{\bg}(-v,u),
\end{equation*}
then, as a direct application of formula~\eqref{eq:Trans_Explicit}, the transvectant $\hh=\trans{\ff}{\bg}{r}$ satisfies
\begin{equation*}
  \hh(u,v)=\overline{\hh}(-v,u).
\end{equation*}
Therefore, the invariants $J_{k}$, produced by the transvectant process, satisfy the following fundamental property
\begin{equation}\label{eq:real-polynomial}
  J(w)\in \RR[W],\quad \text{if} \quad w \in W.
\end{equation}
It remains to show that $\set{J_{1},\dotsc , J_{N}}$ is also a minimal integrity basis for the real $\SO(3,\RR)$-representation $W$, which is the object of the following lemma.

\begin{lem}\label{lem:Real_Integrity_Basis}
  Let $\set{J_{1},\dotsc ,J_{N}}$ be a minimal integrity basis for the complex $\SO(3,\CC)$-representation $W^{\CC}$ such that each polynomial $J_{k} \in \CC[W^{\CC}]$ satisfies~\eqref{eq:real-polynomial}. Then $\set{J_{1},\dotsc ,J_{N}}$ is a minimal integrity basis for the real $\SO(3,\RR)$-representation $W$.
\end{lem}

\begin{proof}
  We will use the following classical result. Let $J \in \RR[W]$, then $J$ is $\SO(3,\RR)$-invariant if and only if its holomorphic extension to $W^{\CC}$ is $\SO(3,\CC)$-invariant. This can easily been checked using the fact that these groups are connected and thus that the assertion needs only to be verified at the level of the Lie algebras. In other words, we have to check that $dJ.\xi_{W}(w) = 0$ for all $\xi \in \mathfrak{so}(3, \RR)$ and all $w\in W$ if and only if $dJ.\xi_{W^{\CC}}(\widetilde{w}) = 0$ for all $\xi \in \mathfrak{so}(3, \CC)$ and all $\widetilde{w}\in W^{\CC}$. Here $\xi_{W}$ and $\xi_{W^{\CC}}$ denote the induced action on the respective Lie algebras and $dJ$ is the differential of $J$. Therefore, if $J \in \RR[W]^{\SO(3,\RR)}$, we can find a polynomial $P\in \CC[X_{1},\dotsc,X_{N}]$ such that
  \begin{equation*}
    J(\widetilde{w}) = P(J_{1}(\widetilde{w}),\dotsc,J_{N}(\widetilde{w})), \qquad \widetilde{w} \in W^{\CC}.
  \end{equation*}
  But $J$, as well as all the $J_{k}$, satisfy~\eqref{eq:real-polynomial}. Hence
  \begin{equation*}
    \overline{J}(w) = J(w), \qquad \overline{J_{k}}(w) = J_{k}(w), \qquad \forall w \in W, \qquad k = 1, \dotsc ,N,
  \end{equation*}
  where $\overline{P}$ is the polynomial defined by taking all conjugate coefficients of $P$. We get thus
  \begin{equation*}
    \overline{J}(w) = \bar{P}(\overline{J_{1}}(w),\dotsc ,\overline{J_{N}}(w)) =\overline{P}(J_{1}(w), \dotsc ,J_{N}(w)) = J(w),
  \end{equation*}
  for all $w\in W$. Therefore, if we set $R := \frac{1}{2}(P+\overline{P})$, then $R\in \RR[X_{1},\dotsc,X_{N}]$ and
  \begin{equation*}
    J(w) = R(J_{1}(w), \dotsc ,J_{N}(w)),
  \end{equation*}
  which shows that $\set{J_{1},\dotsc ,J_{N}}$ is an integrity basis for $W$. If it was not minimal, we would have for instance
  \begin{equation*}
    J_{1}(w) = Q(J_{2}(w), \dotsc ,J_{N}(w)),\quad \forall w\in W.
  \end{equation*}
  Such an identity would then also hold for all $\widetilde{w}\in W^{\CC}$, which would lead to a contradiction.
\end{proof}

\section{Explicit computations}
\label{sec:explicit-computations}

A minimal integrity basis for the space of binary forms $\Sn{4} \oplus \Sn{4} \oplus \Sn{8}$ was computed for the first time in~\cite{Oli2014}. We will use these results, together with an explicit harmonic decomposition $\bC=(\lambda,\mu,\ba,\bb,\bD)$, as detailed in section~\ref{sec:covariants}, to produce a minimal integrity basis for the full elasticity tensor $\bC$. Using the explicit isomorphism $\phi^{*}$, defined in Theorem~\ref{thm:explicit-isomorphism}, between $\Hn{n}(\CC^{3})$ and $\Sn{2n}$, we introduce the following binary forms:
\begin{equation*}
  \hh:=\phi^{*}(\ba)\in \Sn{4},\quad \kk:=\phi^{*}(\bb)\in \Sn{4},\quad \ff:=\phi^{*}(\bD)\in \Sn{8}.
\end{equation*}

As described in section~\ref{subsec:Gordan-algorithm}, it is necessary to compute first a generating set for the \emph{covariant algebras} of $\Sn{8}$ and $\Sn{4} \oplus \Sn{4}$ to compute a generating set for the \emph{invariant algebra} of $\Sn{4} \oplus \Sn{4} \oplus \Sn{8}$. A minimal covariant basis for $\cov{\Sn{8}}$ (see~\cite{Cro2002,LR2012}) is provided in Table~\ref{tab:Cov-S8}, and the covariants are denoted by $\ff_{n}$ ($n=1,\dotsc,69$). A minimal covariant basis for $\cov{\Sn{4} \oplus \Sn{4}}$ is provided in Table~\ref{tab:Cov-S4S4}, and the covariants are denoted by $\hh_{n}$ ($n=1,\dotsc,28$).

A minimal integrity basis for $\HT{4}$ is provided by the 9 invariants in Proposition~\ref{prop:inv-H4}. They correspond to the nine covariants of order $0$: $\ff_{2}$, $\ff_{6}$, $\ff_{14}$, $\ff_{24}$, $\ff_{35}$, $\ff_{44}$, $\ff_{52}$, $\ff_{59}$, $\ff_{64}$ for $\Sn{8}$ in Table~\ref{tab:Cov-S8}.

A minimal integrity basis for $\HT{2}\oplus\HT{2}$ is provided by the 8 invariants in Proposition~\ref{prop:inv-H2H2} (among them, the four simple invariants $\tr(\ba^{2})$, $\tr(\ba^{3})$, $\tr(\bb^{2})$, $\tr(\bb^{3})$). These eight invariants correspond to the covariants of order $0$: $\hh_{3}$, $\hh_{4}$, $\hh_{5}$, $\hh_{11}$, $\hh_{12}$, $\hh_{13}$, $\hh_{14}$, $\hh_{23}$ for $\Sn{4} \oplus \Sn{4}$ in Table~\ref{tab:Cov-S4S4}.

To complete this basis, we have to add twice (for $(\bD,\ba)$ and $(\bD,\bb)$), the 52 joint invariants for $\HT{4} \oplus \HT{2}$ from Table~\ref{tab:S8S4}, where we have introduced the notations:
\begin{equation*}
  \hh_{2,4}:=\trans{\hh}{\hh}{2}\in \Sn{2},\quad \hh_{3,6}:=\trans{\hh}{\hh_{2,4}}{1}\in \Sn{3},
\end{equation*}
and the 174 joint invariants of $\HT{4} \oplus \HT{2} \oplus \HT{2}$ from Table~\ref{tab:S8S4S4}. Note that, in these tables, appear only invariants depending \emph{really} on $(\ff, \hh)$ in the first case, and $(\ff, \hh, \kk)$ in the second case. Thus simple invariants and invariants depending only on $(\ff, \hh)$, $(\ff, \kk)$ or $(\hh, \kk)$ (in the second case) are omitted. We obtain this way $9 + 8 + 2 \times 52 + 174 = 295$ invariants, to which we must add the fundamental invariants $(\lambda, \mu)$ for $\HT{0} \oplus \HT{0}$ to get the 297 invariants of Theorem~\ref{thm:main-result}.

\begin{rem}
  An integrity basis of $299$ invariants was produced in~\cite{Oli2014}. As noticed by Reynald Lercier, this basis was not minimal. Indeed, a degree 11 joint invariant in $\inv{\Sn{8}\oplus\Sn{4}}$ (which needs to be counted twice for our purpose) was superfluous. This mistake has been corrected in~\cite{Oli2016}.
\end{rem}

\begin{table}[p]
  \centering
  \begin{tabular}{|l||l|l||l||l|l|}
    \hline
    Cov.       & Transvectant                     & (d,o)    & Cov.       & Transvectant                       & (d,o)    \\
    \hline
    $\ff_1$    & $\ff$                            & $(1,8)$  & $\ff_{36}$ & $\trans{\ff_{33}}{\ff}{8}$         & $(6,2)$  \\
    $\ff_2$    & $\trans{\ff}{\ff}{8}$            & $(2,0)$  & $\ff_{37}$ & $\trans{\ff_{33}}{\ff}{7}$         & $(6,4)$  \\						
    $\ff_3$    & $\trans{\ff}{\ff}{6}$            & $(2,4)$  & $\ff_{38}$ & $\trans{\ff_{32}}{\ff}{7}$         & $(6,4)$  \\
    $\ff_4$    & $\trans{\ff}{\ff}{4}$            & $(2,8)$  & $\ff_{39}$ & $\trans{\ff_{34}}{\ff}{8}$         & $(6,6)$  \\
    $\ff_5$    & $\trans{\ff}{\ff}{2}$            & $(2,12)$ & $\ff_{40}$ & $\trans{\ff_{33}}{\ff}{6}$         & $(6,6)$  \\
    $\ff_6$    & $\trans{\ff_{4}}{\ff}{8}$        & $(3,0)$  & $\ff_{41}$ & $\trans{\ff_{32}}{\ff}{6}$         & $(6,6)$  \\
    $\ff_7$    & $\trans{\ff_{5}}{\ff}{8}$        & $(3,4)$  & $\ff_{42}$ & $\trans{\ff_{34}}{\ff}{7}$         & $(6,8)$  \\
    $\ff_8$    & $\trans{\ff_{5}}{\ff}{7}$        & $(3,6)$  & $\ff_{43}$ & $\trans{\ff_{34}}{\ff}{6}$         & $(6,10)$ \\
    $\ff_9$    & $\trans{\ff_{5}}{\ff}{6}$        & $(3,8)$  & $\ff_{44}$ & $\trans{\ff_{7}^{2}}{\ff}{8}$      & $(7,0)$  \\
    $\ff_{10}$ & $\trans{\ff_{5}}{\ff}{5}$        & $(3,10)$ & $\ff_{45}$ & $\trans{\ff_{43}}{\ff}{8}$         & $(7,2)$  \\
    $\ff_{11}$ & $\trans{\ff_{5}}{\ff}{4}$        & $(3,12)$ & $\ff_{46}$ & $\trans{\ff_{42}}{\ff}{7}$         & $(7,2)$  \\
    $\ff_{12}$ & $\trans{\ff_{5}}{\ff}{3}$        & $(3,14)$ & $\ff_{47}$ & $\trans{\ff_{43}}{\ff}{7}$         & $(7,4)$  \\
    $\ff_{13}$ & $\trans{\ff_{5}}{\ff}{1}$        & $(3,18)$ & $\ff_{48}$ & $\trans{\ff_{42}}{\ff}{6}$         & $(7,4)$  \\
    $\ff_{14}$ & $\trans{\ff_{9}}{\ff}{8}$        & $(4,0)$  & $\ff_{49}$ & $\trans{\ff_{43}}{\ff}{6}$         & $(7,6)$  \\
    $\ff_{15}$ & $\trans{\ff_{11}}{\ff}{8}$       & $(4,4)$  & $\ff_{50}$ & $\trans{\ff_{42}}{\ff}{5}$         & $(7,6)$  \\
    $\ff_{16}$ & $\trans{\ff_{10}}{\ff}{7}$       & $(4,4)$  & $\ff_{51}$ & $\trans{\ff_{41}}{\ff}{4}$         & $(7,6)$  \\
    $\ff_{17}$ & $\trans{\ff_{12}}{\ff}{8}$       & $(4,6)$  & $\ff_{52}$ & $\trans{\ff_{7}\ff_{16}}{\ff}{8}$  & $(8,0)$  \\
    $\ff_{18}$ & $\trans{\ff_{12}}{\ff}{7}$       & $(4,8)$  & $\ff_{53}$ & $\trans{\ff_{51}}{\ff}{6}$         & $(8,2)$  \\
    $\ff_{19}$ & $\trans{\ff_{13}}{\ff}{8}$       & $(4,10)$ & $\ff_{54}$ & $\trans{\ff_{50}}{\ff}{6}$         & $(8,2)$  \\
    $\ff_{20}$ & $\trans{\ff_{12}}{\ff}{6}$       & $(4,10)$ & $\ff_{55}$ & $\trans{\ff_{51}}{\ff}{5}$         & $(8,4)$  \\
    $\ff_{21}$ & $\trans{\ff_{13}}{\ff}{7}$       & $(4,12)$ & $\ff_{56}$ & $\trans{\ff_{50}}{\ff}{5}$         & $(8,4)$  \\
    $\ff_{22}$ & $\trans{\ff_{13}}{\ff}{6}$       & $(4,14)$ & $\ff_{57}$ & $\trans{\ff_{51}}{\ff}{4}$         & $(8,6)$  \\
    $\ff_{23}$ & $\trans{\ff_{13}}{\ff}{4}$       & $(4,18)$ & $\ff_{58}$ & $\trans{\ff_{50}}{\ff}{4}$         & $(8,6)$  \\
    $\ff_{24}$ & $\trans{\ff_{3}^{2}}{\ff}{8}$    & $(5,0)$  & $\ff_{59}$ & $\trans{\ff_{15}\ff_{16}}{\ff}{8}$ & $(9,0)$  \\
    $\ff_{25}$ & $\trans{\ff_{20}}{\ff}{8}$       & $(5,2)$  & $\ff_{60}$ & $\trans{\ff_{58}}{\ff}{6}$         & $(9,2)$  \\
    $\ff_{26}$ & $\trans{\ff_{21}}{\ff}{8}$       & $(5,4)$  & $\ff_{61}$ & $\trans{\ff_{57}}{\ff}{6}$         & $(9,2)$  \\
    $\ff_{27}$ & $\trans{\ff_{20}}{\ff}{7}$       & $(5,4)$  & $\ff_{62}$ & $\trans{\ff_{16}\ff_{17}}{\ff}{8}$ & $(9,2)$  \\
    $\ff_{28}$ & $\trans{\ff_{22}}{\ff}{8}$       & $(5,6)$  & $\ff_{63}$ & $\trans{\ff_{58}}{\ff}{5}$         & $(9,4)$  \\
    $\ff_{29}$ & $\trans{\ff_{21}}{\ff}{7}$       & $(5,6)$  & $\ff_{64}$ & $\trans{\ff_{17}\ff_{25}}{\ff}{8}$ & $(10,0)$ \\
    $\ff_{30}$ & $\trans{\ff_{22}}{\ff}{7}$       & $(5,8)$  & $\ff_{65}$ & $\trans{\ff_{17}\ff_{27}}{\ff}{8}$ & $(10,2)$ \\
    $\ff_{31}$ & $\trans{\ff_{23}}{\ff}{8}$       & $(5,10)$ & $\ff_{66}$ & $\trans{\ff_{17}\ff_{26}}{\ff}{8}$ & $(10,2)$ \\
    $\ff_{32}$ & $\trans{\ff_{22}}{\ff}{6}$       & $(5,10)$ & $\ff_{67}$ & $\trans{\ff_{27}\ff_{29}}{\ff}{8}$ & $(11,2)$ \\
    $\ff_{33}$ & $\trans{\ff_{21}}{\ff}{5}$       & $(5,10)$ & $\ff_{68}$ & $\trans{\ff_{27}\ff_{28}}{\ff}{8}$ & $(11,2)$ \\
    $\ff_{34}$ & $\trans{\ff_{23}}{\ff}{6}$       & $(5,14)$ & $\ff_{69}$ & $\trans{\ff_{29}\ff_{38}}{\ff}{8}$ & $(12,2)$ \\
    $\ff_{35}$ & $\trans{\ff_{3}\ff_{7}}{\ff}{8}$ & $(6,0)$  &            &                                    &          \\
    \hline
  \end{tabular}
  \caption{A minimal covariant basis for $\Sn{8}$.}
  \label{tab:Cov-S8}
\end{table}

\begin{table}[p]
  \centering
  \begin{tabular}{|l||l|l||l||l|l|}
    \hline
    Cov.       & Transvectant              & $(d_{1},d_{2},o)$ & Cov.       & Transvectant                   & $(d_{1},d_{2},o)$ \\ \hline
    $\hh_{1}$  & $\hh$                     & $(1,0,4)$         & $\hh_{15}$ & $\trans{\hh}{\hh_{8}}{3}$      & $(1,2,2)$         \\
    $\hh_{2}$  & $\kk$                     & $(0,1,4)$         & $\hh_{16}$ & $\trans{\kk}{\hh_{7}}{3}$      & $(2,1,2)$         \\
    $\hh_{3}$  & $\trans{\hh}{\hh}{4}$     & $(2,0,0)$         & $\hh_{17}$ & $\trans{\hh}{\hh_{8}}{2}$      & $(1,2,4)$         \\
    $\hh_{4}$  & $\trans{\kk}{\kk}{4}$     & $(0,2,0)$         & $\hh_{18}$ & $\trans{\kk}{\hh_{7}}{2}$      & $(2,1,4)$         \\
    $\hh_{5}$  & $\trans{\hh}{\kk}{4}$     & $(1,1,0)$         & $\hh_{19}$ & $\trans{\hh}{\hh_{7}}{1}$      & $(3,0,6)$         \\
    $\hh_{6}$  & $\trans{\hh}{\kk}{3}$     & $(1,1,2)$         & $\hh_{20}$ & $\trans{\kk}{\hh_{8}}{1}$      & $(0,3,6)$         \\
    $\hh_{7}$  & $\trans{\hh}{\hh}{2}$     & $(2,0,4)$         & $\hh_{21}$ & $\trans{\hh}{\hh_{8}}{1}$      & $(1,2,6)$         \\
    $\hh_{8}$  & $\trans{\kk}{\kk}{2}$     & $(0,2,4)$         & $\hh_{22}$ & $\trans{\kk}{\hh_{7}}{1}$      & $(2,1,6)$         \\
    $\hh_{9}$  & $\trans{\hh}{\kk}{2}$     & $(1,1,4)$         & $\hh_{23}$ & $\trans{\hh_{7}}{\hh_{8}}{4}$  & $(2,2,0)$         \\
    $\hh_{10}$ & $\trans{\hh}{\kk}{1}$     & $(1,1,6)$         & $\hh_{24}$ & $\trans{\hh_{7}}{\hh_{8}}{3}$  & $(2,2,2)$         \\
    $\hh_{11}$ & $\trans{\hh}{\hh_{7}}{4}$ & $(3,0,0)$         & $\hh_{25}$ & $\trans{\hh_{19}}{\kk}{4}$     & $(3,1,2)$         \\
    $\hh_{12}$ & $\trans{\kk}{\hh_{8}}{4}$ & $(0,3,0)$         & $\hh_{26}$ & $\trans{\hh}{\hh_{20}}{4}$     & $(1,3,2)$         \\
    $\hh_{13}$ & $\trans{\hh}{\hh_{8}}{4}$ & $(1,2,0)$         & $\hh_{27}$ & $\trans{\hh^{2}}{\hh_{20}}{6}$ & $(2,3,2)$         \\
    $\hh_{14}$ & $\trans{\kk}{\hh_{7}}{4}$ & $(2,1,0)$         & $\hh_{28}$ & $\trans{\hh_{19}}{\kk^{2}}{6}$ & $(3,2,2)$         \\ \hline
  \end{tabular}
  \caption{A minimal covariant basis for $\Sn{4}\oplus\Sn{4}$.}
  \label{tab:Cov-S4S4}
\end{table}

\begin{table}[p]
  \centering
  \begin{tabular}{|l|p{11cm}|}
    \hline
    Degree 3  & $\trans{\ff_{3}}{\hh}{4}$,
    $\trans{\ff_{1}}{\hh^{2}}{8}$
    \\ \hline
    Degree 4  & $\trans{\ff_{1}}{\hh\cdot\hh_{2,4}}{8}$,
    $\trans{\ff_4}{\hh^{2}}{8}$,
    $\trans{\ff_{3}}{\hh_{2,4}}{4}$,
    $\trans{\ff_7}{\hh}{4}$
    \\ \hline
    Degree 5  & $\trans{\ff_{1}}{\hh_{2,4}^{2}}{8}$,
    $\trans{\ff_4}{\hh\cdot\hh_{2,4}}{8}$,
    $\trans{\ff_5}{\hh^{3}}{12}$,
    $\trans{\ff_7}{\hh_{2,4}}{4}$,
    $\trans{\ff_9}{\hh^{2}}{8}$,
    $\trans{\ff_{15}}{\hh}{4}$,
    $\trans{\ff_{16}}{\hh}{4}$
    \\ \hline
    Degree 6  & $\trans{\ff_4}{\hh_{2,4}^{2}}{8}$,
    $\trans{\ff_5}{\hh^{2}\cdot\hh_{2,4}}{12}$,
    $\trans{\ff_{11}}{\hh^{3}}{12}$,
    $\trans{\ff_9}{\hh\cdot\hh_{2,4}}{8}$,
    $\trans{\ff_{15}}{\hh_{2,4}}{4}$,
    $\trans{\ff_8}{\hh_{3,6}}{6}$,
    $\trans{\ff_{18}}{\hh^{2}}{8}$,
    $\trans{\ff_{16}}{\hh_{2,4}}{4}$,
    $\trans{\ff_{26}}{\hh}{4}$,
    $\trans{\ff_{27}}{\hh}{4}$
    \\ \hline
    Degree 7  & $\trans{\ff_5}{\hh\cdot\hh_{2,4}^{2}}{12}$,
    $\trans{\ff_{10}}{\hh\cdot\hh_{3,6}}{10}$,
    $\trans{\ff_{11}}{\hh^{2}\cdot\hh_{2,4}}{12}$,
    $\trans{\ff_{18}}{\hh\cdot\hh_{2,4}}{8}$,
    $\trans{\ff_{17}}{\hh_{3,6}}{6}$,
    $\trans{\ff_{21}}{\hh^{3}}{12}$,
    $\trans{\ff_{30}}{\hh^{2}}{8}$,
    $\trans{\ff_{27}}{\hh_{2,4}}{4}$,
    $\trans{\ff_{26}}{\hh_{2,4}}{4}$,
    $\trans{\ff_{37}}{\hh}{4}$, $\trans{\ff_{38}}{\hh}{4}$
    \\ \hline
    Degree 8  & $\trans{\ff_{47}}{\hh}{4}$,
    $\trans{\ff_{48}}{\hh}{4}$,
    $\trans{\ff_{37}}{\hh_{2,4}}{4}$,
    $\trans{\ff_{38}}{\hh_{2,4}}{4}$,
    $\trans{\ff_{42}}{\hh^{2}}{8}$,
    $\trans{\ff_{29}}{\hh_{3,6}}{6}$,
    $\trans{\ff_{30}}{\hh\cdot\hh_{2,4}}{8}$,
    $\trans{\ff_{20}}{\hh\cdot\hh_{3,6}}{10}$,
    $\trans{\ff_{21}}{\hh^{2}\cdot\hh_{2,4}}{12}$,
    $\trans{\ff_{11}}{\hh\cdot\hh_{2,4}^{2}}{12}$
    \\ \hline
    Degree 9  & $\trans{\ff_8^{2}}{\hh^{3}}{12}$,
    $\trans{\ff_{48}}{\hh_{2,4}}{4}$,
    $\trans{\ff_{47}}{\hh_{2,4}}{4}$,
    $\trans{\ff_{55}}{\hh}{4}$,
    $\trans{\ff_{56}}{\hh}{4}$
    \\ \hline
    Degree 10 & $\trans{\ff_{56}}{\hh_{2,4}}{4}$,
    $\trans{\ff_{63}}{\hh}{4}$
    \\ \hline
    Degree 11 & $\trans{\ff_{25}^{2}}{\hh}{4}$
    \\ \hline
  \end{tabular}
  \caption{Joint invariants for $\Sn{8} \oplus \Sn{4}$.}
  \label{tab:S8S4}
\end{table}

\begin{table}[p]
  \centering
  \begin{tabular}{|l|p{11cm}|}
    \hline
    Degree 3  & $\trans{\ff_{1}}{\hh_{1}\cdot\hh_{2}}{8}$
    \\ \hline
    Degree 4  & $\trans{\ff_{1}}{\hh_{1}\cdot\hh_{8}}{8}$,
    $\trans{\ff_{1}}{\hh_{2}\cdot\hh_{9}}{8}$,
    $\trans{\ff_{1}}{\hh_{2}\cdot\hh_{7}}{8}$,
    $\trans{\ff_{1}}{\hh_{1}\cdot\hh_{9}}{8}$,
    $\trans{\ff_{3}}{\hh_{9}}{4}$,
    $\trans{\ff_{4}}{\hh_{1}\cdot\hh_{2}}{8}$
    \\ \hline
    Degree 5  & $\trans{\ff_{1}}{\hh_{8}\cdot\hh_{9}}{8}$,
    $\trans{\ff_{1}}{\hh_{2}\cdot\hh_{17}}{8}$,
    $\trans{\ff_{1}}{\hh_{7}\cdot\hh_{8}}{8}$,
    $\trans{\ff_{1}}{\hh_{2}\cdot\hh_{18}}{8}$,			
    $\trans{\ff_{1}}{\hh_{9}^{2}}{8}$,
    $\trans{\ff_{1}}{\hh_{7}\cdot\hh_{9}}{8}$,
    $\trans{\ff_{1}}{\hh_{1}\cdot\hh_{18}}{8}$,
    $\trans{\ff_{4}}{\hh_{1}\cdot\hh_{8}}{8}$,
    $\trans{\ff_{4}}{\hh_{2}\cdot\hh_{9}}{8}$,			
    $\trans{\ff_{5}}{\hh_{1}\cdot\hh_{2}^{2}}{12}$,
    $\trans{\ff_{3}}{\hh_{17}}{4}$,
    $\trans{\ff_{4}}{\hh_{2}\cdot\hh_{7}}{8}$,
    $\trans{\ff_{3}}{\hh_{18}}{4}$,
    $\trans{\ff_{4}}{\hh_{1}\cdot\hh_{9}}{8}$,
    $\trans{\ff_{5}}{\hh_{1}^{2}\cdot\hh_{2}}{12}$,
    $\trans{\ff_{9}}{\hh_{1}\cdot\hh_{2}}{8}$,
    $\trans{\ff_{7}}{\hh_{9}}{4}$,
    $\trans{\ff_{8}}{\hh_{10}}{6}$
    \\ \hline
    Degree 6  & $\trans{\ff_{1}}{\hh_{8}\cdot\hh_{17}}{8}$,
    $\trans{\ff_{1}}{\hh_{2}\cdot\hh_{6}^{2}}{8}$,
    $\trans{\ff_{1}}{\hh_{9}\cdot\hh_{17}}{8}$,
    $\trans{\ff_{1}}{\hh_{9}\cdot\hh_{18}}{8}$ 			
    $\trans{\ff_{1}}{\hh_{1}\cdot\hh_{6}^{2}}{8}$,
    $\trans{\ff_{1}}{\hh_{7}\cdot\hh_{18}}{8}$,
    $\trans{\ff_{4}}{\hh_{2}\cdot\hh_{17}}{8}$,
    $\trans{\ff_{5}}{\hh_{1}\cdot\hh_{2}\cdot\hh_{8}}{12}$,
    $\trans{\ff_{4}}{\hh_{8}\cdot\hh_{9}}{8}$,
    $\trans{\ff_{5}}{\hh_{2}^{2}\cdot\hh_{9}}{12}$,
    $\trans{\ff_{4}}{\hh_{2}\cdot\hh_{18}}{8}$,
    $\trans{\ff_{5}}{\hh_{1}\cdot\hh_{2}\cdot\hh_{9}}{12}$,
    $\trans{\ff_{5}}{\hh_{1}^{2}\cdot\hh_{8}}{12}$,
    $\trans{\ff_{4}}{\hh_{9}^{2}}{8}$,
    $\trans{\ff_{4}}{\hh_{7}\cdot\hh_{8}}{8}$,
    $\trans{\ff_{5}}{\hh_{2}^{2}\cdot\hh_{7}}{12}$,
    $\trans{\ff_{5}}{\hh_{1}^{2}\cdot\hh_{9}}{12}$,		
    $\trans{\ff_{4}}{\hh_{7}\cdot\hh_{9}}{8}$,
    $\trans{\ff_{4}}{\hh_{1}\cdot\hh_{18}}{8}$,
    $\trans{\ff_{5}}{\hh_{1}\cdot\hh_{2}\cdot\hh_{7}}{12}$,
    $\trans{\ff_{9}}{\hh_{1}\cdot\hh_{8}}{8}$,		
    $\trans{\ff_{8}}{\hh_{21}}{6}$,
    $\trans{\ff_{10}}{\hh_{2}\cdot\hh_{10}}{10}$,
    $\trans{\ff_{8}}{\hh_{2}\cdot\hh_{6}}{6}$,
    $\trans{\ff_{9}}{\hh_{2}\cdot\hh_{9}}{8}$,
    $\trans{\ff_{11}}{\hh_{1}\cdot\hh_{2}^{2}}{12}$,			
    $\trans{\ff_{11}}{\hh_{1}^{2}\cdot\hh_{2}}{12}$,
    $\trans{\ff_{10}}{\hh_{1}\cdot\hh_{10}}{10}$,
    $\trans{\ff_{9}}{\hh_{2}\cdot\hh_{7}}{8}$,
    $\trans{\ff_{9}}{\hh_{1}\cdot\hh_{9}}{8}$,		
    $\trans{\ff_{8}}{\hh_{1}\cdot\hh_{6}}{6}$,
    $\trans{\ff_{8}}{\hh_{22}}{6}$,
    $\trans{\ff_{16}}{\hh_{9}}{4}$,						
    $\trans{\ff_{17}}{\hh_{10}}{6}$,
    $\trans{\ff_{18}}{\hh_{1}\cdot\hh_{2}}{8}$,
    $\trans{\ff_{15}}{\hh_{9}}{4}$
    \\ \hline
    Degree 7  & $\trans{\ff_{5}}{\hh_{2}^{2}\cdot\hh_{17}}{12}$,
    $\trans{\ff_{5}}{\hh_{1}\cdot\hh_{8}^{2}}{12}$,
    $\trans{\ff_{5}}{\hh_{2}\cdot\hh_{8}\cdot\hh_{9}}{12}$,
    $\trans{\ff_{5}}{\hh_{2}^{2}\cdot\hh_{18}}{12}$,
    $\trans{\ff_{5}}{\hh_{1}\cdot\hh_{8}\cdot\hh_{9}}{12}$,
    $\trans{\ff_{5}}{\hh_{2}\cdot\hh_{7}\cdot\hh_{8}}{12}$,
    $\trans{\ff_{5}}{\hh_{2}\cdot\hh_{9}^{2}}{12}$,
    $\trans{\ff_{5}}{\hh_{1}\cdot\hh_{9}^{2}}{12}$,
    $\trans{\ff_{5}}{\hh_{2}\cdot\hh_{7}\cdot\hh_{9}}{12}$,
    $\trans{\ff_{5}}{\hh_{1}\cdot\hh_{2}\cdot\hh_{18}}{12}$,
    $\trans{\ff_{5}}{\hh_{1}\cdot\hh_{7}\cdot\hh_{8}}{12}$,
    $\trans{\ff_{5}}{\hh_{1}\cdot\hh_{7}\cdot\hh_{9}}{12}$,
    $\trans{\ff_{5}}{\hh_{1}^{2}\cdot\hh_{18}}{12}$,
    $\trans{\ff_{5}}{\hh_{2}\cdot\hh_{7}^{2}}{12}$,
    $\trans{\ff_{10}}{\hh_{2}\cdot\hh_{21}}{10}$,
    $\trans{\ff_{10}}{\hh_{1}\cdot\hh_{20}}{10}$,			
    $\trans{\ff_{11}}{\hh_{2}^{2}\cdot\hh_{9}}{12}$,
    $\trans{\ff_{11}}{\hh_{1}\cdot\hh_{2}\cdot\hh_{8}}{12}$,
    $\trans{\ff_{10}}{\hh_{2}^{2}\cdot\hh_{6}}{10}$,			
    $\trans{\ff_{12}}{\hh_{2}^{2}\cdot\hh_{10}}{14}$,
    $\trans{\ff_{10}}{\hh_{1}\cdot\hh_{2}\cdot\hh_{6}}{10}$,
    $\trans{\ff_{11}}{\hh_{1}^{2}\cdot\hh_{8}}{12}$,
    $\trans{\ff_{10}}{\hh_{1}\cdot\hh_{21}}{10}$,			
    $\trans{\ff_{10}}{\hh_{2}\cdot\hh_{22}}{10}$,
    $\trans{\ff_{12}}{\hh_{1}\cdot\hh_{2}\cdot\hh_{10}}{14}$,
    $\trans{\ff_{11}}{\hh_{1}\cdot\hh_{2}\cdot\hh_{9}}{12}$,
    $\trans{\ff_{11}}{\hh_{2}^{2}\cdot\hh_{7}}{12}$,
    $\trans{\ff_{9}}{\hh_{9}^{2}}{8}$, 						
    $\trans{\ff_{10}}{\hh_{1}\cdot\hh_{22}}{10}$,
    $\trans{\ff_{11}}{\hh_{1}\cdot\hh_{2}\cdot\hh_{7}}{12}$,
    $\trans{\ff_{11}}{\hh_{1}^{2}\cdot\hh_{9}}{12}$,
    $\trans{\ff_{10}}{\hh_{1}^{2}\cdot\hh_{6}}{10}$,
    $\trans{\ff_{10}}{\hh_{2}\cdot\hh_{19}}{10}$,
    $\trans{\ff_{12}}{\hh_{1}^{2}\cdot\hh_{10}}{14}$,
    $\trans{\ff_{21}}{\hh_{1}\cdot\hh_{2}^{2}}{12}$,
    $\trans{\ff_{18}}{\hh_{2}\cdot\hh_{9}}{8}$,
    $\trans{\ff_{17}}{\hh_{21}}{6}$,
    $\trans{\ff_{17}}{\hh_{2}\cdot\hh_{6}}{6}$,
    $\trans{\ff_{20}}{\hh_{2}\cdot\hh_{10}}{10}$,
    $\trans{\ff_{19}}{\hh_{2}\cdot\hh_{10}}{10}$,
    $\trans{\ff_{18}}{\hh_{1}\cdot\hh_{8}}{8}$,
    $\trans{\ff_{17}}{\hh_{1}\cdot\hh_{6}}{6}$,
    $\trans{\ff_{18}}{\hh_{1}\cdot\hh_{9}}{8}$,
    $\trans{\ff_{17}}{\hh_{22}}{6}$,
    $\trans{\ff_{20}}{\hh_{1}\cdot\hh_{10}}{10}$,
    $\trans{\ff_{21}}{\hh_{1}^{2}\cdot\hh_{2}}{12}$,
    $\trans{\ff_{19}}{\hh_{1}\cdot\hh_{10}}{10}$,
    $\trans{\ff_{18}}{\hh_{2}\cdot\hh_{7}}{8}$,
    $\trans{\ff_{29}}{\hh_{10}}{6}$,
    $\trans{\ff_{30}}{\hh_{1}\cdot\hh_{2}}{8}$,
    $\trans{\ff_{26}}{\hh_{9}}{4}$,
    $\trans{\ff_{27}}{\hh_{9}}{4}$,
    $\trans{\ff_{28}}{\hh_{10}}{6}$
    \\ \hline
    Degree 8  & $\trans{\ff_{37}}{\hh_{9}}{4}$,
    $\trans{\ff_{38}}{\hh_{9}}{4}$,
    $\trans{\ff_{40}}{\hh_{10}}{6}$,
    $\trans{\ff_{41}}{\hh_{10}}{6}$,
    $\trans{\ff_{42}}{\hh_{1}\cdot\hh_{2}}{8}$,
    $\trans{\ff_{29}}{\hh_{21}}{6}$,
    $\trans{\ff_{30}}{\hh_{1}\cdot\hh_{8}}{8}$,
    $\trans{\ff_{30}}{\hh_{2}\cdot\hh_{9}}{8}$,
    $\trans{\ff_{31}}{\hh_{2}\cdot\hh_{10}}{10}$,
    $\trans{\ff_{32}}{\hh_{2}\cdot\hh_{10}}{10}$,
    $\trans{\ff_{33}}{\hh_{2}\cdot\hh_{10}}{10}$,
    $\trans{\ff_{29}}{\hh_{22}}{6}$,
    $\trans{\ff_{30}}{\hh_{1}\cdot\hh_{9}}{8} $,
    $\trans{\ff_{30}}{\hh_{2}\cdot\hh_{7}}{8}$,
    $\trans{\ff_{31}}{\hh_{1}\cdot\hh_{10}}{10}$,
    $\trans{\ff_{32}}{\hh_{1}\cdot\hh_{10}}{10}$,
    $\trans{\ff_{33}}{\hh_{1}\cdot\hh_{10}}{10}$,
    $\trans{\ff_{20}}{\hh_{2}\cdot\hh_{22}}{10}$,
    $\trans{\ff_{20}}{\hh_{1}\cdot\hh_{2}\cdot\hh_{6}}{10}$,
    $\trans{\ff_{21}}{\hh_{1}^{2}\cdot\hh_{8}}{12}$,
    $\trans{\ff_{21}}{\hh_{1}\cdot\hh_{2}\cdot\hh_{9}}{12}$,
    $\trans{\ff_{21}}{\hh_{2}^{2}\cdot\hh_{7}}{12}$,
    $\trans{\ff_{22}}{\hh_{1}\cdot\hh_{2}\cdot\hh_{10}}{14}$,
    $\trans{\ff_{20}}{\hh_{1}\cdot\hh_{22}}{10}$
    $\trans{\ff_{20}}{\hh_{1}^{2}\cdot\hh_{6}}{10}$,
    $\trans{\ff_{21}}{\hh_{1}^{2}\cdot\hh_{9}}{12}$,
    $\trans{\ff_{21}}{\hh_{1}\cdot\hh_{2}\cdot\hh_{7}}{12}$,
    $\trans{\ff_{22}}{\hh_{1}^{2}\cdot\hh_{10}}{14}$,
    $\trans{\ff_{11}}{\hh_{2}\cdot\hh_{7}\cdot\hh_{9}}{12}$,
    $\trans{\ff_{12}}{\hh_{1}^{2}\cdot\hh_{2}\cdot\hh_{6}}{14}$,
    $\trans{\ff_{13}}{\hh_{1}^{2}\cdot\hh_{2}\cdot\hh_{10}}{18}$,
    $\trans{\ff_{11}}{\hh_{2}\cdot\hh_{7}^{2}}{12}$,
    $\trans{\ff_{12}}{\hh_{1^{3}}\cdot\hh_{6}}{14}$,
    $\trans{\ff_{13}}{\hh_{1^{3}}\cdot\hh_{10}}{18}$,
    $\trans{\ff_{11}}{\hh_{2}\cdot\hh_{9}^{2}}{12}$,
    $\trans{\ff_{12}}{\hh_{1}\cdot\hh_{2}^{2}\cdot\hh_{6}}{14}$,
    $\trans{\ff_{13}}{\hh_{1}\cdot\hh_{2}^{2}\cdot\hh_{10}}{18}$,
    $\trans{\ff_{20}}{\hh_{2}\cdot\hh_{21}}{10}$,
    $\trans{\ff_{20}}{\hh_{2}^{2}\cdot\hh_{6}}{10}$,
    $\trans{\ff_{21}}{\hh_{1}\cdot\hh_{2}\cdot\hh_{8}}{12}$,
    $\trans{\ff_{21}}{\hh_{2}^{2}\cdot\hh_{9}}{12}$,
    $\trans{\ff_{22}}{\hh_{2}^{2}\cdot\hh_{10}}{14}$,
    $\trans{\ff_{11}}{\hh_{2}\cdot\hh_{8}\cdot\hh_{9}}{12}$,
    $\trans{\ff_{12}}{\hh_{2^{3}}\cdot\hh_{6}}{14}$,
    $\trans{\ff_{13}}{\hh_{2^{3}}\cdot\hh_{10}}{18}$
    \\ \hline
    Degree 9  & $\trans{\ff_{1}\cdot\ff_{25}}{\hh_{2}\cdot\hh_{10}}{10}$,
    $\trans{\ff_{43}}{\hh_{2}\cdot\hh_{10}}{10}$,
    $\trans{\ff_{8}^{2}}{\hh_{1}\cdot\hh_{2}^{2}}{12}$,
    $\trans{\ff_{1}\cdot\ff_{25}}{\hh_{1}\cdot\hh_{10}}{10}$,
    $\trans{\ff_{8}^{2}}{\hh_{1}^{2}\cdot\hh_{2}}{12}$,
    $\trans{\ff_{43}}{\hh_{1}\cdot\hh_{10}}{10}$,
    $\trans{\ff_{3}\cdot\ff_{25}}{\hh_{10}}{6}$,
    $\trans{\ff_{51}}{\hh_{10}}{6}$,
    $\trans{\ff_{48}}{\hh_{9}}{4}$,
    $\trans{\ff_{47}}{\hh_{9}}{4}$
    \\ \hline
    Degree 10 & $\trans{\ff_{54}}{\hh_{6}}{2}$,
    $\trans{\ff_{56}}{\hh_{9}}{4}$ 										
    \\ \hline
    Degree 11 & $\trans{\ff_{61}}{\hh_{6}}{2}$,
    $\trans{\ff_{62}}{\hh_{6}}{2}$,
    $\trans{\ff_{63}}{\hh_{9}}{4}$								
    \\ \hline
  \end{tabular}
  \caption{Joint invariants for $\Sn{8} \oplus \Sn{4} \oplus \Sn{4}$.}
  \label{tab:S8S4S4}
\end{table}

\clearpage
\appendix

\section{Harmonic polynomials}
\label{sec:harmonic-polynomials}

There is a well-known correspondence between totally symmetric tensors of order $n$ and  homogeneous polynomials of degree $n$ on $\RR^{3}$. Indeed, to each symmetric tensor $\bT \in \mathbb{S}^{n}(\RR^{3})$ corresponds a homogeneous polynomial of degree $n$ given by
\begin{equation*}\label{map:equivariant_Harm}
  \bp(\vv) := \bT(\vv,\dotsc,\vv), \qquad \vv \in \RR^{3}.
\end{equation*}

This correspondence defines a linear isomorphism $\psi$ between the tensor space $\ST{n}(\RR^{3})$ of totally symmetric tensors of order $n$ on $\RR^{3}$ and the polynomial space $\RR_{n}[\RR^{3}]$ of homogeneous polynomials of degree $n$ on $\RR^{3}$. The inverse $\bT = \psi^{-1}(\bp)$ can be recovered by \emph{polarization}.
More precisely, the expression
\begin{equation*}
  \bp(t_{1}\vv_{1} + \dotsc + t_{n}\vv_{n})
\end{equation*}
is a homogeneous polynomial in the variables $t_{1},\dotsc,t_n$ and we get
\begin{equation*}
  \bT(\vv_{1},\dotsc,\vv_{n}) = \frac{1}{n!} \left.\frac{\partial^{n}}{\partial t_{1} \dotsb \partial t_{n}}\right|_{t_{1} = \dotsb = t_{n} = 0}\bp(t_{1}\vv_{1} + \dotsc + t_{n}\vv_{n}).
\end{equation*}

The rotation group $\SO(3,\RR)$ acts on the polynomial space $\RR_{n}[\RR^{3}]$, by the rule
\begin{equation*}
  (g\cdot \bp)(\vv) := \bp(g^{-1}\cdot \vv), \qquad g \in \SO(3,\RR),
\end{equation*}
and the linear isomorphism $\psi$ is moreover \emph{equivariant}, meaning that
\begin{equation*}
  \psi(g \cdot \bT) = g \cdot \psi(\bT).
\end{equation*}
In other words, the following diagram commutes for $g \in \SO(3,\RR)$:
\begin{equation*}
  \xymatrix{
    \ST{n}(\RR^{3}) \ar[d]_{g} \ar[rr]^{\psi} & & \RR_{n}[\RR^{3}] \ar[d]^{g} \\
    \ST{n}(\RR^{3}) \ar[rr]_{\psi} & &\RR_{n}[\RR^{3}]
  }
\end{equation*}

The sub-space $\HT{n}(\RR^{3})$ of harmonic tensors corresponds under the isomorphism $\psi$ to the sub-space $\Hn{n}(\RR^{3})$ of harmonic polynomials (polynomials with vanishing Laplacian) in $\RR_{n}[\RR^{3}]$ (an other model for irreducible $\SO(3,\RR)$-representations). More precisely, if $\triangle$ denotes the \emph{Laplacian operator} and $\bp = \psi(\bT)$, we get
\begin{equation*}
  \triangle \bp = n(n-1) \psi(\tr \bT).
\end{equation*}
Thus, \emph{totally symmetric tensors with vanishing trace} correspond precisely to \emph{harmonic polynomials}. This justifies the denomination of \emph{harmonic tensors} for elements of $\HT{n}(\RR^{3})$.

The following lemma gives the precise decomposition of a homogeneous polynomial of degree $n$ (and thus of a totally symmetric tensor of order $n$) into its irreducible components called \emph{harmonic components}.

\begin{lem}\label{lem:Proj_Harmonique}
  Let $\bp \in \RR_{n}[\RR^{3}]$ and $r=[n/2]$. Then we have
  \begin{equation}\label{eq:Harm_Decomp}
    \bp = \bh_{0} + \qq\bh_{1} + \dotsb + \qq^{r}\bh_{r}, \quad \bh_{k}\in \Hn{n-2k}(\RR^{3}),
  \end{equation}
  where $\qq(x,y,z) = x^{2} + y^{2} + z^{2}$ and $\bh_{k}$ is a harmonic polynomial defined recursively by
  \begin{equation*}
    \bh_{r} = \left\{
    \begin{array}{ll}
      \displaystyle{\frac{1}{(2r+1)!}}\triangle^{r} \bp,              & \hbox{if $n$ is even;} \\
      \displaystyle{\frac{3!\times (r+1)}{(2r+3)!}}\triangle^{r} \bp, & \hbox{if $n$ is odd,}
    \end{array}
    \right.
  \end{equation*}
  and for $k<r$
  \begin{equation*}
    \bh_{k} = \mu(k) \triangle^{k}\left(\bp - \sum_{j=k+1}^{r} \qq^{j} \bh_{j}\right),\quad \mu(k) := \frac{(2n-4k+1)!(n-k)!}{(2n-2k+1)!k!(n-2k)!}.
  \end{equation*}
\end{lem}

\begin{proof}
  The lemma results from the following observation. If $\bh \in \Hn{n}(\RR^{3})$, then
  \begin{equation*}
    \triangle^{k}(\qq^{k}\bh) = \lambda_{k}(n)\bh, \qquad 1 \le k \le r,
  \end{equation*}
  where
  \begin{equation*}
    \lambda_{k}(n) = \frac{(2(n+k)+1)!\,k!\,n!}{(2n+1)!\,(n+k)!}.
  \end{equation*}
  We obtain first $\bh_{r}$ and then, recursively, $\bh_{r-1}, \dotsc,\bh_{0}$.
\end{proof}

\section{The Cartan map}
\label{sec:Cartan-map}

Let $\SL(2,\CC)$ be the group of matrices $\gamma \in M_{2}(\CC)$ of determinant $1$. Its Lie algebra $\slc(2,\CC)$ is the vector space of matrices $M \in M_{2}(\CC)$ with vanishing trace. This space is of (complex) dimension 3. An explicit isomorphism between $\CC^{3}$ and $\slc(2,\CC)$ is given by
\begin{equation}\label{eq:spinor-representation}
  \xx := (x,y,z) \mapsto M(\xx) = \left(
  \begin{array}{cc}
    -z   & x+iy \\
    x-iy & z    \\
  \end{array}
  \right).
\end{equation}

The adjoint action of $\SL(2,\CC)$ on $\slc(2,\CC)$
\begin{equation*}
  \Ad_{\gamma} : M \mapsto \gamma M \gamma^{-1}, \qquad M \in \slc(2,\CC), \, \gamma \in \SL(2,\CC),
\end{equation*}
preserves the complex quadratic form
\begin{equation*}
  \det M = -(x^{2} + y^{2} + z^{2})
\end{equation*}
on $\slc(2,\CC)$ and it can be checked, moreover, that $\det \Ad_{\gamma} = 1$ for all $\gamma \in \SL(2,\CC)$. Therefore, we deduce a group morphism
\begin{equation}\label{eq:universal-cover}
  \pi : \gamma \mapsto Ad_{\gamma}, \qquad \SL(2,\CC) \to \SO(3,\CC),
\end{equation}
where
\begin{equation*}
  \SO(3,\CC) := \set{g \in M_{3}(\CC); \; g^{t}g = I \quad \text{and} \quad \det g = 1}.
\end{equation*}

\begin{rem}
  $\SL(2,\CC)$ is a two-fold cover of $\SO(3,\CC)$. It is in fact its \emph{universal cover} and is called the \emph{spinor group} of $\SO(3,\CC)$.
\end{rem}

\begin{rem}
  We can restrict this morphism $\pi$ to the subgroup
  \begin{equation*}
    \SU(2) := \set{\gamma \in \SL(2,\CC);\; \bar{\gamma}^{t} \gamma = I}.
  \end{equation*}
  Note that the Lie algebra of $\SU(2)$ corresponds to matrices $M$ in \eqref{eq:spinor-representation}, with $x,y,z$ purely imaginary, and thus $\pi(\SU(2)) = \SO(3,\RR)$. $\SU(2)$ is moreover the spinor group of $\SO(3,\RR)$.
\end{rem}

Consider now the skew-symmetric 2-form on $\CC^{2}$
\begin{equation*}
  \omega(\bxi_{1},\bxi_{2}) := \det (\bxi_{1},\bxi_{2}), \qquad \bxi_{1},\bxi_{2} \in \CC^{2},
\end{equation*}
and define the mapping
\begin{equation*}
  \bxi \mapsto \bxi^{\omega}, \qquad \CC^{2} \to \left(\CC^{2}\right)^{*},
\end{equation*}
where $\left(\CC^{2}\right)^{*}$ is the dual of the vector space $\CC^{2}$, and
\begin{equation*}
  \bxi_{1}^{\omega}(\bxi_{2}) := \omega(\bxi_{1},\bxi_{2}).
\end{equation*}
In the canonical basis of $\CC^{2}$ and its dual basis, we get
\begin{equation*}
  \bxi = \left(
  \begin{array}{c}
    u \\
    v \\
  \end{array}
  \right),
  \qquad
  \bxi^{\omega} = \left(
  \begin{array}{cc}
    -v & u \\
  \end{array}
  \right).
\end{equation*}
The \emph{Cartan map} is defined as
\begin{equation}\label{eq:Cartan-map}
  \phi : \bxi = \left(
  \begin{array}{c}
    u \\
    v \\
  \end{array}
  \right)
  \mapsto
  \bxi\bxi^{\omega} = \left(
  \begin{array}{cc}
    -uv    & u^{2} \\
    -v^{2} & uv    \\
  \end{array}
  \right).
\end{equation}
Note that $\tr \bxi\bxi^{\omega} = 0$ and $\det \bxi\bxi^{\omega} = 0$, and we obtain, thus, a mapping
\begin{equation*}
  \phi : \CC^{2} \to \slc(2,\CC),
\end{equation*}
whose image lies inside the isotropic cone
\begin{equation*}
  C := \set{M \in \slc(2,\CC) ;\; \det M = 0}.
\end{equation*}
In the complex coordinates $(x,y,z)$ of $\slc(2,\CC)$ introduced in~\eqref{eq:spinor-representation}, the Cartan map is given by
\begin{equation*}
  x = \frac{u^{2} - v^{2}}{2}, \qquad y = \frac{u^{2} + v^{2}}{2i}, \qquad z = uv.
\end{equation*}
We deduce from this explicit expression that $\phi$ is surjective onto $C$ and that each point in the cone $C$ has exactly two pre-images $\bxi$ and $-\bxi$.

\begin{rem}
  Note that the mapping~\eqref{eq:spinor-representation} from $\CC^{3}$ to $\mathrm{End}(\CC^{2})$ satisfies:
  \begin{equation*}
    M(\xx)M(\yy) + M(\yy)M(\xx) = 2(\xx \cdot \yy)I,
  \end{equation*}
  where
  \begin{equation*}
    \xx \cdot \yy := x_{1}x_{2} + y_{1}y_{2} + z_{1}z_{2}, \qquad \xx, \yy \in \CC^{3},
  \end{equation*}
  and induces an irreducible representation of the complex Clifford algebra $\mathrm{Cl}_{3}(\CC)$ into $\mathrm{End}(\CC^{2})$. The Cartan map~\eqref{eq:Cartan-map} was introduced by Cartan in 1913 (see~\cite{Car1981}) and rediscovered later by physicists (see for instance~\cite{Bac1970,Hob1955}). The two pre-images, $\bxi$ and $-\bxi$, of a matrix $M \in C$ by $\phi$ are called \emph{spinors} and are extremely important mathematical objects in quantum mechanics. Note that if we write a vector $\vv = (x,y,z) \in \CC^{3}$ as $\vv = \vv_{1} + i\vv_{2}$, where $\vv_{j}\in \RR^{3}$, the condition $\vv \in C$, \textit{i.e}:
  \begin{equation*}
    x^{2} + y^{2} + z^{2} = 0
  \end{equation*}
  means that
  \begin{equation*}
    \norm{\vv_{1}} = \norm{\vv_{2}}, \quad \text{and} \quad \vv_{1} \cdot \vv_{2} = 0.
  \end{equation*}
\end{rem}

\begin{lem}\label{lem:Cartan-equiv}
  The Cartan map
  \begin{equation*}
    \phi : \CC^{2} \to \slc(2,\CC),
  \end{equation*}
  is $\SL(2,\CC)$-equivariant, i.e.
  \begin{equation*}
    \phi(\gamma \cdot \bxi) = \Ad_{\gamma}\phi(\bxi), \qquad \bxi \in \CC^{2}, \, \gamma \in \SL(2,\CC).
  \end{equation*}
\end{lem}

\begin{proof}
  We have
  \begin{equation*}
    \phi(\gamma \cdot \bxi) = (\gamma \bxi)(\gamma \bxi)^{\omega}
  \end{equation*}
  but
  \begin{equation*}
    (\gamma \bxi)^{\omega} = \bxi^{\omega}\gamma^{-1},
  \end{equation*}
  thus
  \begin{equation*}
    \phi(\gamma \cdot \bxi) = (\gamma \bxi)(\bxi^{\omega}\gamma^{-1}) = \Ad_{\gamma} \bxi\bxi^{\omega} = \Ad_{\gamma}\phi(\bxi).
  \end{equation*}
\end{proof}

\end{document}